\documentclass[10pt,twocolumn,twoside]{IEEEtran}

\usepackage{amsmath}
\usepackage{tabularx}
\usepackage{graphicx}
\usepackage{subfigure}
\usepackage{amssymb}
\usepackage{amsthm}
\usepackage{epstopdf}
\usepackage{color}
\usepackage{algorithm,algpseudocode}
\newtheorem{theorem}{Theorem}
\newtheorem{corollary}{Corollary}

\newtheorem{lemma}{Lemma}
\newtheorem{definition}{Definition}
\newtheorem{assumption}{Assumption}

\newtheorem{example}{Example}
\usepackage{url} 
\usepackage{cite}

\usepackage{fancyhdr}
\usepackage{lipsum}

\begin{document}

\title{
QoS Based Contract Design for Profit Maximization in IoT-Enabled Data Markets}

\author{
        Juntao~Chen,~\IEEEmembership{Member,~IEEE}, Junaid Farooq,~\IEEEmembership{Member,~IEEE}
       and Quanyan~Zhu,~\IEEEmembership{Senior Member,~IEEE}
\thanks{This work was supported in part by the National Science Foundation (NSF) under Grants ECCS-1847056, CNS-2027884, and BCS-2122060.}
\thanks{Juntao Chen is with the Department of Computer and Information Sciences,  Fordham University,  New York,  NY 10023 USA. E-mail: jchen504@fordham.edu.}
\thanks{Junaid Farooq is with the Department of Electrical \& Computer Engineering, College of Engineering and Computer Science, University of Michigan-Dearborn, Dearborn, MI 48128 USA. E-mail: mjfarooq@umich.edu.}
\thanks{Quanyan Zhu is with the Department of Electrical and Computer Engineering, Tandon School of Engineering, New York University, Brooklyn, NY, 11201 USA. E-mail:  qz494@nyu.edu.}}

\maketitle

\begin{abstract}
The massive deployment of Internet of Things (IoT) devices, including sensors and actuators, is ushering in smart and connected communities of the future. The massive deployment of Internet of Things (IoT) devices, including sensors and actuators, is ushering in smart and connected communities of the future. The availability of real-time and high-quality sensor data is crucial for various IoT applications, particularly in healthcare, energy, transportation, etc. However, data collection may have to be outsourced to external service providers (SPs) due to cost considerations or lack of specialized equipment. Hence, the data market plays a critical role in such scenarios where SPs have different quality levels of available data, and IoT users have different application-specific data needs. The pairing between data available to the SP and users in the data market requires an effective mechanism design that considers the SPs' profitability and the quality-of-service (QoS) needs of the users. We develop a generic framework to analyze and enable such interactions efficiently, leveraging tools from contract theory and mechanism design theory. It can enable and empower emerging data sharing paradigms such as Sensing-as-a-Service (SaaS). The contract design creates a pricing structure for on-demand sensing data for IoT users. By considering a continuum of user types, we capture a diverse range of application requirements and propose optimal pricing and allocation rules that ensure QoS provisioning and maximum profitability for the SP. Furthermore, we provide analytical solutions for fixed distributions of user types to analyze the developed approach. For comparison, we consider the benchmark case assuming complete information of the user types and obtain optimal contract solutions. Finally, a case study based on the example of virtual reality application delivered using unmanned aerial vehicles (UAVs) is presented to demonstrate the efficacy of the proposed contract design framework.

\end{abstract}

\section{Introduction}

The Internet of things (IoT) applications rely heavily on sensed data from a multitude of sources resulting in powerful and intelligent applications based on sensor fusion and machine learning. For instance, smart and connected communities, industrial automation, smart grid all rely on reliable and high quality data for automated decision-making~\cite{zanella2014internet}. To fulfil the data needs of intelligence-based IoT applications, the sensing and data acquisition tasks can be outsourced to professional service providers (SPs) in the data market~\cite{xiong2017connectivity}. It results in cost effective data collection for IoT applications, wider choice of sensing data, and on-demand service delivery to users.
For example, in an intelligent transportation network, vehicles can choose the services to communicate with roadside infrastructures that belong to sensing SP for exchanging various types of data related to applications such as GPS navigation, parking, and highway tolls inquiries. etc. Another potential scenario is UAV-enabled virtual reality (VR) experiences~\cite{UAV-VR}. As shown in Fig. \ref{VR}, the UAVs managed by the SP capture 3D images of areas that users are interested in, and send them to the remote users via cloud servers and communication networks. These images can be of varying quality and resolution suited for a range of different user types.
Therefore, the service interactions between the users and the sensing SP requires a formal contract design, in which IoT users make subscription contracts with the SP to obtain (real-time) sensor data according to specific mission requirements~\cite{perera2014sensing}. 

\begin{figure}[!t]
\begin{centering}
\includegraphics[width=0.95\columnwidth]{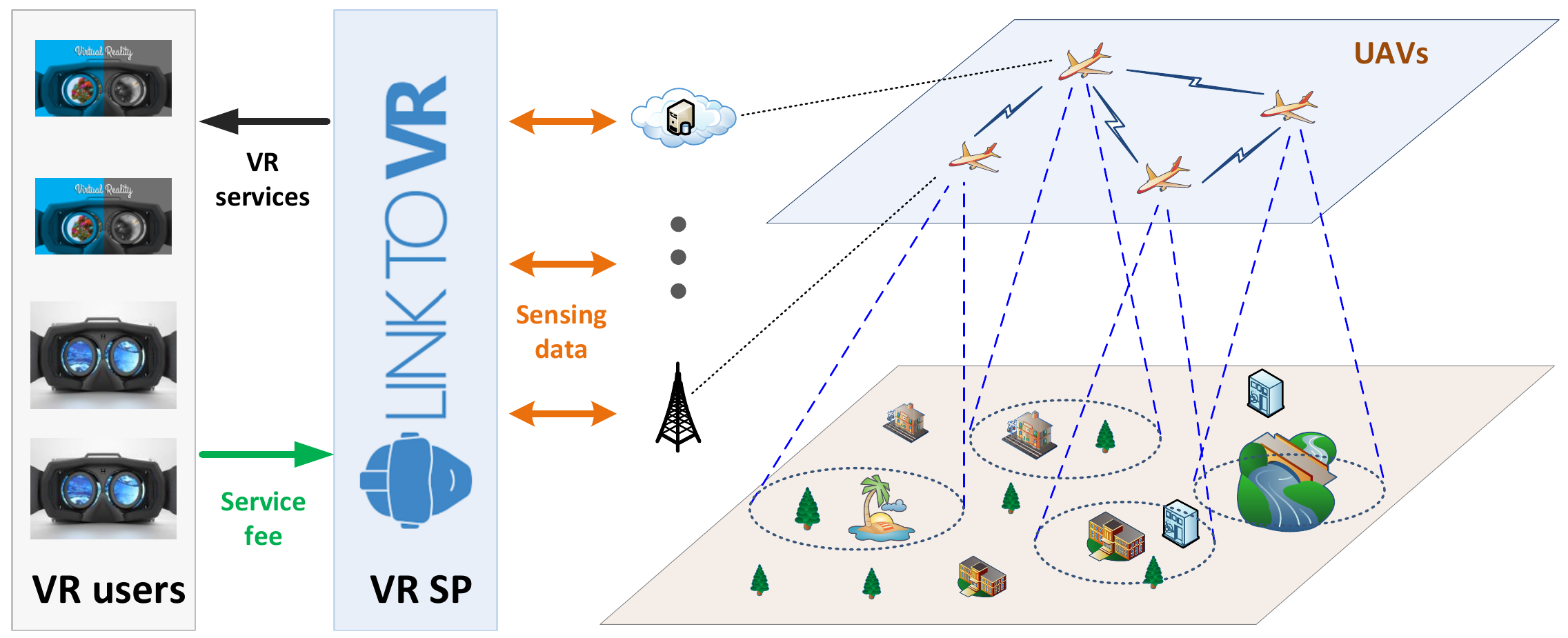}
\par\end{centering}
\caption{\label{VR}
In the UAV-enabled VR applications, the UAVs capture views of the areas of interest. The collected data are aggregated in the cloud, which is managed by the VR SP, and then sent to the remote users. The real-time 3D information delivery is useful in applications such as remote monitoring, navigation, and entertainment. Based on the application, VR users have different QoS requirements and pay different service fees.}
\end{figure}

Depending on the particular application, IoT users have different requirements on the quality of data provided by the sensing SP. Note that provisioning of high-quality sensing data demands high-level of investment in terms of equipment deployment, maintenance, technical support, and data processing from the SP. In the UAV-enabled VR, users may require different levels of quality-of-service (QoS) in terms of the transmission delay and resolution of the images. Therefore, users with different QoS needs can be classified into different types\footnote{The user types can also be interpreted as the importance of tasks to the users respectively, ranging from non mission-critical to mission-critical ones.}.  
The sensing SP aims to maximize its revenue and minimize the service costs jointly by delivering on-demand sensing services. In contrast, the user's goal is to choose a service that maximizes its utility. Therefore, there is a need to design efficient contracting strategies between the SP and the users so that sensing technologies can be effectively monetized.
In the proposed contract design framework, the SP needs to design a menu of contracts that specify the sensing price and the QoS offered to each type of user. The optimal contracts yield a matching between the available sensing data and users in the IoT ecosystem that is suitable for both SP and the users.

Due to the large-population feature of users in the massive IoT \cite{semasinghe2017game}, the SP may not be aware of the exact type of each user and may only have high level information on the distribution of user's types (e.g., inferred from historical demand data)\footnote{This asymmetric information assumption also aligns with the fact that the users aim to preserve privacy of their true types.}. Thus, the challenge of contract design lies in the development of an incentive compatible and optimal mechanism for the sensing SP to maximize its payoff by serving IoT users inspite of the incomplete information. To overcome this obstacle, we propose a market-based pricing contract mechanism for the SaaS model that takes into account incentive compatibility and individual rationality of the users. Specifically, we consider a continuum of user types with a generic probability distribution and design optimal contracts leveraging the Pontryagin maximum principle~\cite{maximum_principle}. 

Under a wide class of probability distributions of user's type, we obtain analytical expression of optimal contracts in which the pricing scheme and the QoS mapping are monotonically increasing with user's types, creating a complete sensing service market with all possible QoS levels. When the probability density function of user's type distribution has a large or sudden decrease around some points, then nondiscriminative pricing phenomenon occurs, which reduces the diversity of service provisions to the IoT users. Specifically, some users choose the same service contract in spite of their heterogeneous types. In addition, nondiscriminative pricing for all customers can occur when the user's types are nested in the lower regime. Hence, in this scenario, the SP should target at the majority in the market to optimize the revenue. For comparison, we study the optimal contracts under complete information and characterize the solution differences.

We illustrate the optimal SaaS mechanism design principles with an application to the UAV-enabled VR. Simulation results show that the SP earns more profit by serving users with relatively stringent service requirements (higher types). However, since the users of lower types constitute most of the market, the SP gains a large proportion of revenue from serving low type users even though their unit benefit is smaller.

The main contributions of this paper are summarized as follows:
\begin{enumerate}
\item We propose a two-sided market-based SaaS contract design for QoS driven data trading between the service provider and users in the IoT ecosystem under asymmetric information.
\item We characterize the solutions of optimal contracts for arbitrary distributions of user types, that yield the best matching between the sensing services and the users leveraging the Pontryagin maximum principle.
\item We show that under the efficient data pricing mechanism, the optimal contracts  either capture the diversity of user types (discriminative pricing) or focus on the majority of user types (nondiscriminative pricing) depending on the users' preferences. 
\item We provide an illustrative example of UAV-enabled VR application to validate and test our proposed contract design. We further provide a comparison between the hidden and full information scenarios in terms of the 
payoff of the SP. 
\end{enumerate}

\subsection{Related Work}
Contract design~\cite{laffont2009theory} has typically been used in operations research with applications to retail, financial markets \cite{fehr2007fairness}, insurances \cite{doherty1993insurance}, supply chains \cite{corbett2000supplier}, etc. With the emergence of IoT and the data markets~\cite{data_market_design_overview}, new service models such as the SaaS are being developed enabling new possibilities such as
resource trading~\cite{resource_trading_iot,contracts_resource_trading_mec}, opportunistic IoT~\cite{contract_opportunistic_iot},
task offloading and outsourcing~\cite{ computation_offloading_trans_scheduling_delay_sensitive}, and performance oriented resource provisioning~\cite{incentive_mechanism_resource_alloc_edge_fog,AoI_latency_contract_telco}. Therefore, there is a need for developing effective contracts~\cite{computational_caching_AR_contracts} and pricing schemes~\cite{optimal_pricing,price_auditing_ridehailing} that incentivize the interactions between users and service providers of data in the IoT ecosystem. The data markets and contract solutions can be implemented using blockchain infrastructure over IoT networks~\cite{blockchain_as_a_service_5G,data_trading_blockchain_IoT,contract_blockchain_IoT}.

A variety of literature is available on using contract theory for incentive mechanism design in wireless communication systems~\cite{zhang2017survey},
tailored for scenarios such as traffic offloading \cite{zhang2015contract,chen2017promoting}, relay selection \cite{hasan2013relay}, spectrum trading \cite{duan2014cooperative,gao2011spectrum}, etc. In \cite{chang2018incentive}, the authors have studied the resource trading process between a mobile virtual wireless network operator and infrastructure providers using a contract. Similar approaches have also been used to facilitate Wi-Fi sharing in crowdsourced wireless community networks \cite{ma2016contract}. Incentive mechanism design has also been received a lot of attention in the next-generation crowdsensing applications. For example, a two-stage Stackelberg game approach has been proposed in \cite{nie2018stackelberg} to design incentive mechanism for the crowdsensing service provider by capturing the participation level of the mobile users. In \cite{xiong2019dynamic}, the authors have investigated the sequential dynamic pricing scheme of a monopoly mobile network operator in the social data market by considering the congestion effects in wireless networks. In \cite{duan2014motivating}, a distributed computing approach is used in crowdsourcing using contracts by focusing on designing a reward-based collaboration mechanism. Contract theory is also leveraged to price the sponsored content in mobile service \cite{xiong2020contract}, where the authors developed a hierarchical game framework to capture the service relationships between the network operator acting as the leader and the content provider and the end users acting as followers.

Our work focuses on establishing a sensing data trading platform enabled by the IoT by considering the user's rationality and market reputation in a holistic manner.
Different from \cite{al2017price} where the authors have focused on designing a pricing mechanism for data delivery in massive IoT from a routing perspective, we address the data pricing problem based on a contract-theoretic approach. Regarding SaaS in the IoT, \cite{al2013priced} has established a public sensing framework for service-based applications in smart cities where the data is provided by a cloud platform. The authors in \cite{kantarci2014trustworthy} have investigated smart phone-based crowdsensing to enhance the public safety via the collected sensing data. In this paper, we use an analytical approach to create an implementable policy framework, focusing on a large-population regime through contract design, which facilitates the realisation of the SaaS paradigm.

We highlight several differences of this work with the literature that uses contracts in various service provisioning applications related to IoT and (wireless) communications. Different from the majority of works (e.g., \cite{zhang2015contract,chen2017promoting,duan2014cooperative,gao2011spectrum,chang2018incentive,ma2016contract,nie2018stackelberg,xiong2019dynamic,duan2014motivating,xiong2020contract}) that have considered finite number of user `types' in the contract formulation, our framework focuses on a large-population regime of IoT users and uses a density function to describe the heterogeneous types of users. The second difference is that our framework considers the reputation of service provisioning through an average QoS constraint. This constraint implicitly improves the inclusion of distinct types of users in the service market. The third difference is on the solution approach used. Instead of solving the problem from an classical optimization angle, this work addresses the problem from an optimal control perspective.

\subsection{Organization of the Paper}
The rest of the paper is organized as follows. Section \ref{formulation} introduces the SaaS framework and formulates the contracting problem. Contract analysis under a class of user's type distributions is presented in Section \ref{optimal_design}. We provide the detailed optimal contract solutions for two special cases in Section \ref{special_case}. Section \ref{Benchmark} investigates the contract design under complete information. Extensions of the contract design to general user's type distributions are presented in Section \ref{extension}. Section \ref{examples} illustrates the obtained results with an application to UAV-based VR, and Section \ref{conclusion} concludes the paper.

\section{System Model and Problem Formulation}\label{formulation}

We consider a pool of IoT users with varying QoS requirements, that are connected to an SP for obtaining sensing data. We assume that the SP has similar sensing data available in a variety of different quality levels. For instance, the same video data can be available in many different pixel resolutions.
Each user obtains a particular quality of data from the SP for its specific mission needs. Depending on the application and quality of data required, the IoT users can be characterized by their `type', denoted by $\delta$. In the following subsections, we provide a description of the different model parameters and an analytical formulation of the optimal contract between an SP and IoT users.

\subsection{User Type and Data Quality}\label{type_data_quality}

Considering a large number of users in a massive IoT setting, each user is characterized by its type
$\delta\in\Delta:=[\underline{\delta},\bar{\delta}]$, which is hidden to the SP, where $\underline{\delta}\geq 0$ and $\bar{\delta}\geq 0$ denote the lower and upper bounds of the parameter, respectively. Here, $\delta$ signifies the importance level of user's task depending on the application needs. Furthermore, considering a large number of possible user types, we assume a continuum of $\delta$ admitting a value from the set $\Delta$. 
The incomplete information of the IoT users to SP implies that the SP does not know the individual attributes of the users. However, the SP may have a broad understanding of the probability distribution of the users. This preserves user's privacy to a certain degree.
Hence, instead of knowing the explicit information of $\delta$, we assume that the sensing SP has knowledge only about the probability density function of the users' type, denoted by $f(\delta)$.

\begin{example}\textbf{Empirical Estimation of User Type Distribution}\\
To design practical contracts for VR services in the IoT, we plot the data of customers' spending preferences on the VR equipment in Fig. \ref{VR_real_data}. The data is adapted from \cite{VR_data}. Since a higher price of VR equipment generally yields a better quality of VR experience, the data can be used to approximate the distribution of customers' types in our VR contract design. Fig. \ref{VR_real_data} depicts five levels of customers' types. Without loss of generality, we can consider their types as type 0, type 1, type 2, type 3, and type 4, respectively, from left to right. In the proposed SaaS framework, we consider an on-demand sensing service provision in a large-population regime. Thus, the customer's type parameter is continuous over a bounded support. Motivated by Fig. \ref{VR_real_data}, we consider the type parameter $\delta$ taking a value from the interval $\Delta:=[\underline{\delta},\bar{\delta}]$, where $\underline{\delta}=0$, $\bar{\delta}=4$, and a larger $\delta$ indicates a higher requirement of VR data quality. This modeling is consistent with the statistics shown in Fig. \ref{VR_real_data}. Furthermore, based on Fig. \ref{VR_real_data}, $\delta$ empirically admits an exponential distribution. Using statistical inference techniques, we can obtain $f(\delta) = 0.952e^{-0.952\delta}$, and $F(\delta) =  1-e^{-0.952\delta}$. Note that these probability density and distribution functions are aligned with the market data in Fig. \ref{VR_real_data}.
\end{example}

\begin{figure}[t]
\begin{centering}
\includegraphics[width=0.95\columnwidth]{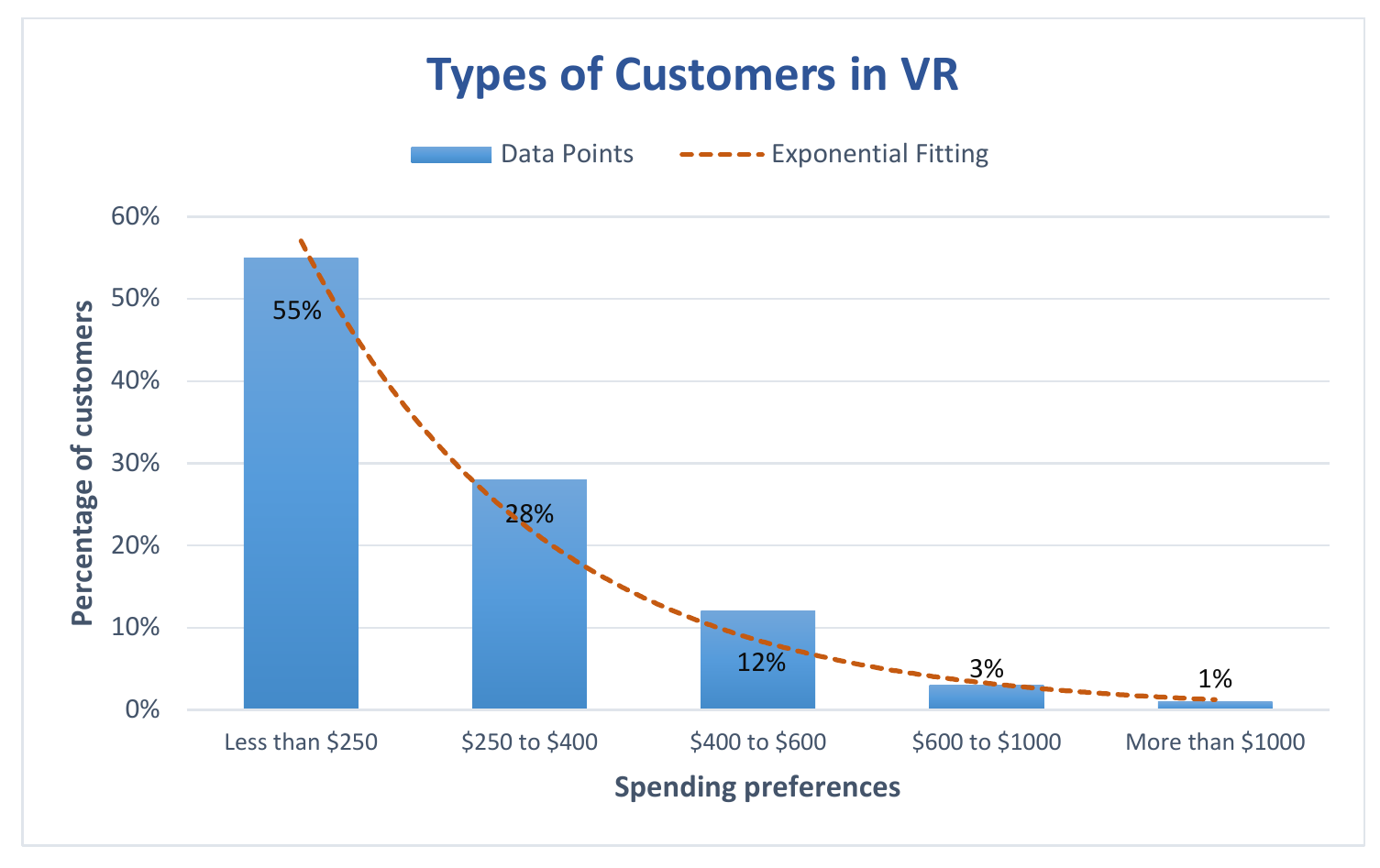}
\par\end{centering}
\caption{\label{VR_real_data}
Customers' spending preferences in the VR headset. Note that a higher price of VR equipment can be interpreted as the customer preferring higher quality of VR experiences. Then, the data yields an empirically exponential distribution of the customers' types in our contract design for VR services.}
\end{figure}

The sensed data available to the SP is characterized by its QoS level, denoted by $q \in \mathbb{R}$, and the corresponding price (payment by the user), denoted by $p \in \mathbb{R}$. The QoS level can be related to a number of specific metrics, such as the pixel density, latency, and jitter in the transmission of sensing data, etc. Note that we consider a continuum of quality levels since a large number of different versions of data are assumed to be available to the SP. In general, we can consider a vectorized $q$, where each element denotes the quality of the corresponding metric. The set $\mathcal{Q}$ denotes the available QoS levels provided by the SP.

\subsection{QoS Provisioning and Profit of SP}

The service relationships between users and SP described above can be naturally captured by a contract-theoretic framework. Specifically, due to the asymmetric information induced by users' hidden type, the SP needs to design a menu of contracts, i.e., $\{q(\delta),p(\delta)\}$ and present it to the users. Each user will then choose one contract that maximizes its payoff. The payoff of the user with type $\delta$, which claims to be of type $\delta'$ (thus receiving contact $\{q(\delta'),p(\delta')\}$) can be computed as
\begin{equation}\label{V_fun}
V(\delta,\delta') = \Phi\left(\delta,q(\delta')\right)-p(\delta'),
\end{equation}
where $V:\Delta\times \Delta\rightarrow \mathbb{R}$, and  $\Phi: \Delta\times \mathcal{Q}\rightarrow \mathbb{R}$. Note that the function $\Phi$ is a measure of the utility of the user. A natural assumption of $\Phi$ is described as follows:

\begin{assumption}\label{A1}
The function $\Phi$ is continuously differentiable and increasing in variables $\delta$ and $q$, i.e., $\frac{\partial \Phi(\delta,q(\delta))}{\partial \delta}>0$ and $\frac{\partial \Phi(\delta,q(\delta))}{\partial q(\delta)}>0$. Also, it satisfies $\frac{\partial{\Phi^2(\delta,q(\delta))}}{\partial q(\delta)\partial{\delta}}> 0$.
\end{assumption}
Assumption \ref{A1} indicates that with a better QoS level, the payoff of user increases. Also, for a given QoS level, the users with a larger type parameter $\delta$ have a higher payoff since their tasks are more mission-critical. Furthermore, for a same amount enhancement of QoS level, the resulting payoff increases for higher type users exceeds the one associated with lower types users.

The function describing the SP's profit obtained by providing a QoS level $q$ to a user of type $\delta$, is defined as
\begin{align}\label{U_SP}
U(\delta)=p(\delta)-C(q(\delta)),
\end{align}
where $C:\mathcal{Q}\rightarrow \mathbb{R}_+$ is the cost of the SP for providing the sensor data. Then,
the expected total payoff of the SP can be expressed as
$
 \int_{\underline{\delta}}^{\bar{\delta}}  (p(\delta)-C(q(\delta))) f(\delta)d\delta,
$ where $f(\delta)$ denotes the density of type $\delta$ users. 
We consider that $f(\delta)$ is strictly greater than 0, i.e., $f(\delta)>0,\ \forall\delta\in[\underline{\delta},\bar{\delta}]$, which holds in the case of massive IoT.

\subsection{Profit Maximizing Contract Problem}
Based on the direct revelation principle  \cite{myerson1979incentive}, it is sufficient for the SP to design/consider contracts in which the users can truthfully select the one that is consistent with their true types; in other words, the users will reveal their types in the selection and do not have incentives to misrepresent their true types. Hence, we characterize the incentive compatibility (IC) and individual rationality (IR) constraints of the users defined as follows.

\begin{definition}[Incentive Compatibility]
A menu of contracts $\{q(\delta),p(\delta)\}$, $\forall \delta$, designed by the sensing SP is incentive compatible if the user of type $\delta$ selects the contract $(q(\delta),p(\delta))$ that maximizes its payoff, i.e.,
\begin{equation}\label{IC}
\Phi(\delta,q(\delta))-p(\delta) \geq \Phi(\delta,q(\delta'))-p(\delta'),\ \forall (\delta,\delta')\in\Delta^2.
\end{equation}
\end{definition}

\begin{definition}[Individual Rationality]
The individual rationality constraint of each user is captured by
\begin{equation}\label{IR}
\Phi(\delta,q(\delta))-p(\delta)\geq 0, \ \forall \delta\in\Delta.
\end{equation}
\end{definition}

To investigate the impact of average sensing QoS level on the optimal contracts, the SP has an additional constraint on the provided QoS to the IoT users as follows:
\begin{equation}\label{reputation}
\int_{\underline{\delta}}^{\bar{\delta}} q(\delta) f(\delta)d\delta\geq \underline{q},
\end{equation}
where the positive constant $\underline{q}$ is the mean/average QoS. The constraint \eqref{reputation} can be interpreted as the \textit{reputation} that the SP aims to build in the sensing service market. Note that the average QoS has been leveraged to guide the optimal decision-making in various applications in literature, such as bandwidth allocation in broadband service provisioning \cite{roy2022achieving}, admission control of services in edge computing \cite{chen2022optimal}, and transmit powers minimization in small cell base stations \cite{lakshminarayana2015transmit}.

The goal of the SP is to jointly determine the pricing scheme $p(\delta)$ and the corresponding service quality $q(\delta)$ that yields the best return. To this end, the SP is required to solve the following optimization problem:
\begin{align*}
\mathrm{(OP):}\  \max_{\{q(\delta),p(\delta)\}}\ & \int_{\underline{\delta}}^{\bar{\delta}} \big (p(\delta)-C(q(\delta))\big) f(\delta)d\delta\\
\mathrm{s.t.}\  \Phi&(\delta,q(\delta))-p(\delta) \geq \Phi(\delta,q(\delta'))-p(\delta'),\\
&\qquad\qquad\qquad\quad\  \forall (\delta,\delta')\in\Delta^2,\ \mathrm{(IC)}\\
&\Phi(\delta,q(\delta))-p(\delta)\geq 0, \ \forall \delta\in\Delta,\ \mathrm{(IR)}\\
&\int_{\underline{\delta}}^{\bar{\delta}} q(\delta) f(\delta)d\delta\geq \underline{q}.\ \mathrm{(Reputation)}
\end{align*}

\section{Analysis and Design of Optimal Contracts}\label{optimal_design}
In this section, we first analyze the formulated problem (OP) in Section \ref{formulation}. Then we design the optimal contracts for SaaS by using Pontryagin maximum principle \cite{kirk2012optimal}.

\subsection{Problem Analysis}\label{sec:prob_analysis}
To solve  problem (OP), one challenge lies in the infinite number of IC and IR constraints in \eqref{IC} and \eqref{IR}, respectively.
To simplify (OP), we first present the following lemma.

\begin{lemma}\label{L1}
Under the condition that $p(\delta)$ and $q(\delta)$ are differentiable, the set of IC constraints \eqref{IC} is equivalent to the local incentive constraints
\begin{equation}\label{IC_differential}
\frac{dp(\delta)}{d\delta} = \frac{\partial{\Phi(\delta,q(\delta))}}{\partial{q(\delta)}}\frac{dq(\delta)}{d\delta},\ \forall \delta\in\Delta,
\end{equation}
and a monotonicity constraint
\begin{equation}\label{monotone}
\frac{dq(\delta)}{d\delta} \geq 0.
\end{equation}
\end{lemma}
\begin{proof}
See Appendix \ref{Appx1}. 
\end{proof}

To facilitate the optimal contract design in Section \ref{optimal_design_results}, we specify some structures of the payoff, $\Phi$ and cost, $C$. The user with mission-critical tasks (higher $\delta$) gains more by receiving better quality of sensor data (higher $q$). Therefore, a reasonable payoff function for type $\delta$ user can be chosen as follows.

\begin{assumption}
The payoff function for type $\delta$ user is considered as
\begin{equation}\label{Phi}
\Phi(\delta,q(\delta)) = \delta q(\delta).
\end{equation}
\end{assumption}

Then, \eqref{IC_differential} can be simplified as $\frac{dp(\delta)}{d\delta} = \delta \frac{dq(\delta)}{d\delta}$. Note that the payoff function \eqref{Phi} is not necessary linear. The only requirement of $\Phi$ is to satisfy  Assumption \ref{A1}. The analysis of optimal contract design in the following still holds for a general $\Phi$.
Similarly, to obtain analytical results of optimal contracts, one of the cost functions of sensing SP is chosen as follows.
\begin{assumption}
The cost function of sensing SP is
\begin{equation}\label{cost_SP}
C(q(\delta))=\sigma\left(\exp({aq(\delta)})-1\right),
\end{equation}
where $\sigma>0$ is a normalizing constant trading off between the sensing service costs and the revenue, and $a>0$ is a sensitivity constant, indicating that the marginal cost of the sensor data is increasing with its quality.
\end{assumption}

\begin{corollary}\label{C1}
Based on Lemma \ref{L1} and \eqref{Phi}, the IC constraints in (OP) can be represented as
\begin{equation}\label{d_V}
\frac{dV}{d\delta} = q(\delta),
\end{equation}
together with the monotonicity constraint \eqref{monotone}.
\end{corollary}
\begin{proof}
Based on \eqref{V_fun}, we have 
$
\frac{dV}{d\delta} = \frac{d\Phi}{d\delta} - \frac{dp(\delta)}{d\delta}
$. Then, using $\frac{d\Phi}{d\delta} = \delta \frac{dq(\delta)}{d\delta}+q(\delta)$ and $\frac{dp(\delta)}{d\delta} = \delta \frac{dq(\delta)}{d\delta}$ yield the result. 
\end{proof}

Note that Corollary \ref{C1} indicates that the payoff of the user is monotonically increasing with the type $\delta$. Therefore, the IR constraint can be simplified as
$
\Phi(\underline{\delta},q(\underline{\delta}))-p(\underline{\delta})\geq 0.
$
Indeed, the IR constraint is binding under the optimal contracts for type $\underline{\delta}$ users, i.e., 
\begin{equation}
\Phi(\underline{\delta},q(\underline{\delta}))-p(\underline{\delta})= 0.
\end{equation}
Otherwise, the sensing SP can earn more profits by increasing the price $p(\underline{\delta})$ for serving the type $\underline{\delta}$ users. 

The reputation constraint \eqref{reputation} essentially divides the problem analysis in two regimes: whether the constraint is binding at the optimal solution or not. Denote by $\{p^*(\delta),q^*(\delta)\}$ the optimal solution to (OP). When $\underline{q}$ is relatively large, then it is possible that $\int_{\underline{\delta}}^{\bar{\delta}} q^*(\delta) f(\delta)d\delta = \underline{q}$, since the SP has no incentive to provide a QoS $q(\delta)$ with $q(\delta)>q^*(\delta)$ which decreases the objective value. The inequality $\int_{\underline{\delta}}^{\bar{\delta}} q^*(\delta) f(\delta)d\delta> \underline{q}$ could happen when $\underline{q}$ is relatively small. Thus, there exists a threshold of $\underline{q}$ above which \eqref{reputation} is binding and below which is non-binding at the optimum. In the case of $\int_{\underline{\delta}}^{\bar{\delta}} q^*(\delta) f(\delta)d\delta> \underline{q}$ which indicates that \eqref{reputation} is inactive, the optimal solution $\{p^*(\delta),q^*(\delta)\}$ to (OP) will be the same as the one to (OP) without considering the reputation constraint. To this end, we have the following approach to address (OP) for given $\underline{q}$. First, we solve (OP) without considering the reputation constraint. If the obtained solution satisfies the reputation constraint, then it is optimal to (OP). Otherwise, we replace the constraint \eqref{reputation} in (OP) by
\begin{equation}\label{reputation_2}
\int_{\underline{\delta}}^{\bar{\delta}} q(\delta) f(\delta)d\delta= \underline{q},
\end{equation}
as the reputation constraint holds as an equality at the optimal contract design in this regime. Solving (OP) without incorporating the reputation constraint is a classical optimal contract design problem. In this work, we focus on developing a systematic approach to address the second case where \eqref{reputation_2} is considered in the constraints.

\textit{Remark:} The reputation constraint implicitly penalizes the SP for not serving IoT users with low valuation (the users of lower types). The reason is that not serving users is equivalent to providing zero quality service with zero cost which indeed decreases the average QoS. This is different from the setup in the classical contract design in which the SP only serves the consumers with positive valuations (e.g., based on the metric called \textit{virtual valuation} $\delta - \frac{1-F(\delta)}{f(\delta)}$). Our model aims to serve all users including those of low types that might not contribute to the SP's profit. Hence, the proposed framework with the reputation constraint has the capability to enhance the accessibility and affordability of the service to all users.

\subsection{Optimal Contract Solution}\label{optimal_design_results}
Based on \eqref{V_fun}, we obtain $p(\delta) = \Phi(\delta,q(\delta))-V(\delta) $, where we suppress the notations $q$ and $p$ in $V$. 
This is reasonable since the payoff of an IoT user depends on its type when the contract is designed.
Then, by regarding $V(\delta)$ as a decision variable instead of $p(\delta)$, the problem can be rewritten as
\begin{align*}
\mathrm{(OP'):}&\\
 \max_{\{q(\delta),V(\delta)\}}\ & \int_{\underline{\delta}}^{\bar{\delta}} \big (\Phi(\delta,q(\delta))-V(\delta) -C(q(\delta))\big) f(\delta)d\delta\\
\mathrm{s.t.}\ & \frac{dV}{d\delta} = q(\delta),\ \frac{dq(\delta)}{d\delta} \geq 0,\  V(\underline{\delta})= 0,\\
&\int_{\underline{\delta}}^{\bar{\delta}} q(\delta) f(\delta)d\delta= \underline{q}.
\end{align*}

$(\mathrm{OP}')$ can be regarded as an optimal control problem. Specifically, by following the notations in control theory, we denote $u(\delta) = q(\delta)$ by the control variable and $x_1(\delta) = V(\delta)$ by the state variable. Then, we obtain $\dot{x}_1 = u(\delta)$ with the initial value $x_1(\underline{\delta}) = 0$. The control input admits an increasing property with the type parameter $\delta$, i.e., $\dot{u}(\delta)\geq 0$.

The remaining difficulty in solving $(\mathrm{OP}')$ lies in the reputation constraint. To facilitate the design of the optimal control strategy, we introduce a new state variable $x_2(\delta)$ satisfying $\dot{x}_2(\delta) = u(\delta)f(\delta)$. Therefore, the reputation constraint can be replaced by
\begin{equation}\label{x_2}
\dot{x}_2(\delta)=u(\delta)f(\delta),
\end{equation}
with boundary values: $x_2(\bar{\delta}) = \underline{q}$ and $x_2(\underline{\delta})=0$.

For clarity, we present the problem $(\mathrm{OP}')$ with new notations as follows:
\begin{align*}
\mathrm{(OP''):}\\
\max_{\{u(\delta),\mathbf{x}(\delta)\}}\ & \int_{\underline{\delta}}^{\bar{\delta}} \big (\Phi(\delta,u(\delta)) - x_1(\delta) -C(u(\delta))\big) f(\delta)d\delta\\
\mathrm{s.t.}\quad & \dot{x}_1(\delta) = u(\delta),\ x_1(\underline{\delta})= 0,\\
&\dot{x}_2(\delta)=u(\delta)f(\delta),\ x_2(\bar{\delta}) = \underline{q},\ x_2(\underline{\delta})=0,\\
&\dot{u}(\delta) \geq 0.
\end{align*}
where $\mathbf{x}=[x_1,x_2]^T$.
Next, by defining $\boldsymbol{\lambda}= [\lambda_1,\lambda_2]^T$,  the Hamiltonian of $(\mathrm{OP}'')$ can be expressed as
\begin{equation}\label{Hamilton}
\begin{split}
H&(\mathbf{x}(\delta), u(\delta), \boldsymbol{\lambda}(\delta),\delta) = \big[\Phi(\delta,u(\delta)) - x_1(\delta)\\
&-C(u(\delta))\big] f(\delta) + \lambda_1(\delta) u(\delta)  +\lambda_2(\delta) u(\delta)f(\delta),
\end{split}
\end{equation}
where $\lambda_1$ and $\lambda_2$ are costate variables corresponding to \eqref{d_V} and \eqref{x_2}, respectively.

By using the Pontryagin maximum principle \cite{kirk2012optimal}, we can obtain the optimal solution $(\mathbf{x}^*(\delta),u^*(\delta))$  by solving the following Hamilton system:
\begin{align}
H&(\mathbf{x}^*(\delta), u^*(\delta), \boldsymbol{\lambda}^*(\delta),\delta) \geq H(\mathbf{x}^*(\delta), u(\delta), \boldsymbol{\lambda}^*(\delta),\delta),\label{H_1}\\
\dot{x}_1^* &= \frac{\partial H(\mathbf{x}^*(\delta), u^*(\delta), \boldsymbol{\lambda}^*(\delta),\delta)}{\partial \lambda_1(\delta)} = u^*(\delta),\label{H_2}\\
\dot{x}_2^* &= \frac{\partial H(\mathbf{x}^*(\delta), u^*(\delta), \boldsymbol{\lambda}^*(\delta),\delta)}{\partial \lambda_2(\delta)} = u^*(\delta)f(\delta),\label{H_3}\\
\dot{\lambda}_1^* &= -\frac{\partial H(\mathbf{x}^*(\delta), u^*(\delta), \boldsymbol{\lambda}^*(\delta),\delta)}{\partial x_1(\delta)} = f(\delta),\label{H_4}\\
\dot{\lambda}_2^* &= -\frac{\partial H(\mathbf{x}^*(\delta), u^*(\delta), \boldsymbol{\lambda}^*(\delta),\delta)}{\partial x_2(\delta)} = 0,\label{H_5}\\
&\lambda_1(\bar{\delta})  =0,\label{H_6}\\
&\lambda_2(\bar{\delta})\ \mathrm{is\ a\ constant}.\label{H_7}
\end{align}
Note that \eqref{H_6} and \eqref{H_7} are boundary conditions. Specifically, the initial state of $x_1$ is fixed, and we only have freedom in specifying boundary condition at the terminal time. Then, the corresponding costate variable $\lambda_1$ at the time $\bar{\delta}$ should equal to the derivative of the terminal payoff with respect to the state $x_1$ at $\bar{\delta}$ based on the maximum principle. Since the objective function in $(\mathrm{OP}'')$ does not include the terminal payoff, then we obtain $\lambda_1(\bar{\delta})  =0$. Similarly, the initial and terminal states of $x_2$ are fixed, and we can specify the boundary condition \eqref{H_7} from \eqref{H_5} in which $\lambda_2$ admits a constant value.

Furthermore, \eqref{H_1} ensures the optimality of control $u^*(\delta)$. Thus, using the first-order condition, \eqref{H_1} can be simplified as
\begin{equation}\label{maximality}
\begin{split}
\frac{\partial H(\mathbf{x}^*(\delta), u(\delta), \boldsymbol{\lambda}^*(\delta),\delta)}{\partial u(\delta)}  = \left(\frac{\partial \Phi(\delta,u(\delta))}{\partial u(\delta)}  - \frac{dC(u(\delta))}{du(\delta)}\right)\\
\cdot f(\delta) + \lambda_1^*(\delta)+\lambda^*_2(\delta)f(\delta)= 0.
\end{split}
\end{equation}

In addition, \eqref{H_4} and \eqref{H_6} indicate that 
\begin{equation}\label{lamb1}
{\lambda}_1^*(\delta) = F(\delta)-1.
\end{equation}
Note that the end-point of $x_2$ is fixed, and hence $\lambda_2(\bar{\delta})$ needs to be determined rather than simply being 0. Based on \eqref{H_5}, we obtain 
\begin{equation}\label{lamb2}
\lambda_2^*(\delta) = \beta,\ \forall \delta\in\Delta, 
\end{equation}
where $\beta$ is a constant to be determined. 

We have obtained the optimal solutions for $\lambda_1^*(\delta)$ and $\lambda_2^*(\delta)$. To design the optimal $u^*(\delta)$, we next focus on the optimality condition \eqref{maximality}. The distribution of user's type can be general, e.g., normal, exponential, or learnt from the historical data.  We first solve $\mathrm{(OP'')}$ without considering the monotonicity constraint $\dot{u}(\delta)\geq 0$. Then, the obtained control $u^{co}(\delta)$ from \eqref{maximality} is a candidate optimal solution. By plugging \eqref{lamb1} and \eqref{lamb2} into \eqref{maximality} and using the defined functions \eqref{Phi} and \eqref{cost_SP}, we obtain $\frac{dC(u^{co}(\delta))}{du(\delta)} - \frac{\partial \Phi(\delta,u(\delta))}{\partial u(\delta)}   =  \frac{F(\delta)-1}{f(\delta)}+\beta$ which leads to
\begin{align}
u^{co}(\delta) = \frac{1}{a}\ln\left(\frac{1}{a\sigma}\left[\frac{F(\delta)-1}{f(\delta)}+ \delta + \beta \right] \right).\label{u_delta}
\end{align}
The second-order condition gives $\frac{\partial^2 H(\mathbf{x}^*(\delta), u(\delta), \boldsymbol{\lambda}^*(\delta),\delta)}{\partial u(\delta)^2} = -\sigma a^2e^{au(\delta)}f(\delta)<0$, and hence $u^{co}(\delta)$ is a maximizer of the Hamiltonian. The maximum principle is a necessary condition for the optimal solution of $\mathrm{(OP'')}$. Then, we further check the sufficient condition for optimality on the maximized Hamiltonian. Specifically, by verifying that the Hamiltonian $H(\mathbf{x}(\delta), u(\delta), \boldsymbol{\lambda}(\delta),\delta)$ is concave in both $\mathbf{x}$ and $u$, the solution $u^{co}(\delta)$ is optimal to $\mathrm{(OP'')}$ without considering the monotonicity constraint. Indeed, based on the Mangasarian sufficiency theorem \cite{mangasarian1966sufficient}, a stronger conclusion is that the obtained control $u^{co}(\delta)$ is the unique optimal solution as the Hamiltonian is strictly concave in $u$. Based on the dynamics in $\mathrm{(OP'')}$, the optimal state trajectory  is also unique.

We next verify whether $u^{co}(\delta)$ satisfying the monotonicity constraint $\dot{u}(\delta)\geq 0$. In \eqref{u_delta}, the CDF $F(\delta)$ is increasing with $\delta$, but the presence of $f(\delta)$ makes the monotonicity of ${u}(\delta)$ unclear. We present the following lemma which can be proved using optimality condition to \eqref{u_delta}.
 
\begin{lemma}\label{Lemma2}
If $2f^2(\delta)+(1-F(\delta))f'(\delta)>0$, then the obtained solution $u^{co}(\delta)$ is optimal. In addition, a decreasing $\frac{1-F(\delta)}{f(\delta)}$ leads to an optimal $u^{co}(\delta)$.
\end{lemma}
\begin{proof}
We need to ensure that \eqref{u_delta} is increasing with $\delta$. The first-order condition of \eqref{u_delta} gives $\left( \frac{F(\delta)-1}{f(\delta)}+\delta+\beta \right)^{-1}\left( \frac{f^2(\delta)-(F(\delta)-1)f'(\delta)}{f^2(\delta)}+1\right)>0$. In the first part, $\beta$ is a constant determined based on $\int_{\underline{\delta}}^{\bar{\delta}} u^{co}(\delta) f(\delta)d\delta=\int_{\underline{\delta}}^{\bar{\delta}} \frac{1}{a}\ln(\frac{1}{a\sigma}[\frac{F(\delta)-1}{f(\delta)}+ \delta + \beta ] )f(\delta)d\delta  = \underline{q}$. The integrand should be well-defined to make the equation satisfied. The existence of such $\beta$ is guaranteed as $\int_{\underline{\delta}}^{\bar{\delta}} \ln(\frac{1}{a\sigma}[\frac{F(\delta)-1}{f(\delta)}+ \delta + \beta ] )f(\delta)d\delta$ is monotonically increasing in $\beta$. Thus, $\frac{F(\delta)-1}{f(\delta)}+\delta+\beta>0$ which is ensured by the choice of $\beta$. Then, we need to have $\frac{f^2(\delta)-(F(\delta)-1)f'(\delta)}{f^2(\delta)}+1>0$ which gives the result. We can also verify that if $\frac{1-F(\delta)}{f(\delta)}$ is decreasing in $\delta$, then $2f^2(\delta)+(1-F(\delta))f'(\delta)>0$ holds which yields the result.
\end{proof}

\textit{Remark:} The distributions of IoT user's type satisfying the condition in Lemma \ref{Lemma2} are quite general, including the uniform, normal and exponential ones. Note that the distributions without a large and sudden decrease in the probability density function (PDF) $f(\delta)$ generally satisfy the condition in Lemma \ref{Lemma2}, and hence \eqref{u_delta} gives the optimal solution. 

Back to \eqref{lamb2}, the constant $\beta$  can be obtained by solving the reputation constraint \eqref{reputation_2}, i.e., 
$
\int_{\underline{\delta}}^{\bar{\delta}} \frac{1}{a}\ln(\frac{1}{a\sigma}[\frac{F(\delta)-1}{f(\delta)}+ \delta + \beta ]) f(\delta)d\delta= \underline{q}.
$
The expression $u^*(\delta)$ characterizes the provided sensing QoS in terms of the user's type. We next focus on obtaining the pricing scheme of the sensing services. To this end, the expression of $x^*_1(\delta)$ becomes critical. Based on \eqref{H_2}, we obtain 
\begin{equation}\label{dot_x}
\dot{x}^*_1(\delta)  = \frac{1}{a}\ln\left(\frac{1}{a\sigma}\left[\frac{F(\delta)-1}{f(\delta)}+ \delta + \beta \right] \right).
\end{equation}
Then, $x_1^*(\delta)$ can be determined by \eqref{dot_x}
and  $x_1^*(\underline{\delta})= 0$. The following Theorem \ref{thm1} explicitly characterizes the optimal contracts in the considered scenario.

\begin{theorem}\label{thm1}
Under the condition $2f^2(\delta)+(1-F(\delta))f'(\delta)>0$ in Lemma \ref{Lemma2}, the optimal contracts $\{q^*(\delta),p^*(\delta)\}$ designed by the SP are as follows:
\begin{equation}\label{thm1_eq}
\begin{split}
q^*(\delta) &= \frac{1}{a}\ln\left(\frac{1}{a\sigma}\left[\frac{F(\delta)-1}{f(\delta)}+ \delta + \beta \right] \right),\\
p^*(\delta) &= \Phi(\delta,	q^*(\delta)) - \phi(\delta) = \delta q^*(\delta)- \phi(\delta),
\end{split}
\end{equation}
where $\beta$ is determined from $\int_{\underline{\delta}}^{\bar{\delta}} q^*(\delta) f(\delta)d\delta= \underline{q}$, and $\dot \phi(\delta) :=  \frac{1}{a}\ln(\frac{1}{a\sigma}[\frac{F(\delta)-1}{f(\delta)}+ \delta + \beta ] )$ with $\phi(\underline{\delta})=0$.
\end{theorem}

\textit{Structure of the optimal contracts:} The $q^*(\delta)$ in \eqref{thm1_eq} can be naturally decomposed into three parts, and each one includes a term $\frac{F(\delta)-1}{f(\delta)}$, $\delta$, and $\beta$, corresponding to the incentives of IoT users, the utility of SP, and the reputation of service provision, respectively. Recall that ${\lambda}_1^*(\delta) = F(\delta)-1$. Thus, the first term quantifying the impact of IC constraint on $q^*(\delta)$ captures the statistics of the IoT user types. The second term including $\delta$ arises from the maximization of objective function of SP which yields him the largest revenue. The third constant term $\beta$ indicates that the sensitivity of reputation constraint is the same for every type of users. This finding is consistent with the fact that the reputation constraint takes the aggregated service provision over all users into account, i.e., the mean QoS. The service pricing function $p^*(\delta)$ is characterized based on $q^*(\delta)$ through relation \eqref{V_fun} and hence has a similar decomposition interpretation as $q^*(\delta)$. In sum, the structure of optimal contracts in Theorem \ref{thm1} incorporates a service payoff maximization term and two adjusting terms for user incentives.

\section{Analytical Results of Special Cases}\label{special_case}
In this section, we present analytical results of optimal contracts for two typical distributions of the user's type.

\subsection{
Uniform User Type Distribution
}
When $\delta$ is uniformly distributed, its PDF and cumulative density function (CDF) admit the forms:
$
f(\delta) = \frac{1}{\bar{\delta}-\underline{\delta}}$ and
$F(\delta) = \frac{\delta - \underline{\delta}}{\bar{\delta} - \underline \delta},\ \delta\in\Delta.
$
Based on Theorem \ref{thm1}, the sensing QoS function is $q^{*}(\delta) = \frac{1}{a}\ln (\frac{1}{a\sigma}[ 2\delta -\bar{\delta} + \beta  ]).$
Due to the logarithm function, $q^*(\delta)$ is nonlinear with $\delta$. In addition, the marginal sensing QoS is decreasing with the IoT user's type. One reason is that increasing sensing QoS is harder in large $q$ regime than its counterpart for the SP. Further, the unknown constant $\beta$ in \eqref{lamb2} can be solved from $\left( \frac{\beta-\bar{\delta}}{2} + \bar{\delta}\right)\ln(\beta+\bar{\delta}) - \left( \frac{\beta-\bar{\delta}}{2} + \underline{\delta}\right) \ln(\beta-\bar{\delta}+2\underline{\delta}) = (a\underline{q}+\ln (a\sigma)+1)(\bar{\delta}-\underline{\delta})$.

The optimal sensing pricing scheme in the contract is characterized in the following corollary.

\begin{corollary}\label{corollary_q_uni}
Under the uniform distribution of the IoT user's type $\delta$, the price of sensing service is equal to
$p^*(\delta) = \delta q^*(\delta) + \frac{\delta-\underline{\delta}}{a(\bar{\delta}-\underline{\delta})}(\ln(a\sigma)+1)- \frac{1}{a(\bar{\delta}-\underline{\delta})}\Big[ \left( \frac{\beta -\bar{\delta}}{2} +\delta \right)\ln(\beta-\bar{\delta}+2\delta) - \left( \frac{\beta -\bar{\delta}}{2} + \underline\delta \right) \ln(\beta-\bar{\delta}+2\underline\delta) \Big]$.
\end{corollary}

\subsection{Exponential User Type Distribution}\label{sec:exp}
When the user's type $\delta$ admits the exponential distribution, then the number of IoT users with mission-critical tasks is less than the ones with nonmission-critical tasks.  Specifically, the PDF and CDF of $\delta$ with rate $\rho$ are equal to
$
f(\delta) = \rho e^{-\rho\delta}$ and $
F(\delta) =  1-e^{-\rho\delta}
$, respectively.
Then, the optimal sensing QoS function has the form
$
q^{*}(\delta) = \frac{1}{a}\ln (\frac{1}{a\sigma}[ \delta -\frac{1}{\rho} + \beta ]),
$
where $\beta$ can be computed from
$
\int_{\underline{\delta}}^{\bar{\delta}}\ln (\frac{1}{a\sigma} [ \delta -\frac{1}{\rho} + \beta  ]) \rho e^{-\rho\delta} d\delta= a \underline{q}.
$

Similar to the uniform distribution scenario, we can characterize the optimal pricing  as follows. 

\begin{corollary}\label{corollary_q_expnential}
Under the exponential distribution of the user's type $\delta$, the optimal pricing of the sensing service in the contract is
$
p^*(\delta) = \delta q^*(\delta) - \frac{1}{a}(\delta+\beta-\frac{1}{\rho}) \ln(\delta+\beta -\frac{1}{\rho}) + \frac{\delta}{a}\left(1+\ln (a\sigma)\right) -\gamma,
$
where the constant $\gamma$ is equal to $\gamma = \frac{\underline{\delta}}{a}(1+\ln (a\sigma)) - \frac{1}{a}(\underline{\delta}+\beta-\frac{1}{\rho}) \ln(\underline{\delta}+\beta -\frac{1}{\rho} )$.
\end{corollary}

We elaborate more on  exponential distribution scenario in case studies in Section \ref{examples}. In other cases with more general distributions of the IoT user's type, we can directly apply Theorem \ref{thm1} to obtain the optimal SaaS contracts. However, note that the support of $f(\delta)$ needs to be consistent with the range of $\delta$. Hence, if a normal distribution is used, it needs to be truncated in order to be compatible with the framework.

\section{Comparison to the Benchmark Scenario}\label{Benchmark}
Under the full information scenario, the sensing SP knows the type of each IoT user. Thus, the IC constraint \eqref{IC} becomes no longer necessary. Then, the optimal contract design problem for SaaS  becomes:
\begin{align*}
\mathrm{(OP-B):}\
 \max_{\{q(\delta),V(\delta)\}}\ & \int_{\underline{\delta}}^{\bar{\delta}} \big (p(\delta) -C(q(\delta))\big) f(\delta)d\delta\\
\mathrm{s.t.}\ 
 V({\delta})&\geq 0,\ \forall \delta,\
\int_{\underline{\delta}}^{\bar{\delta}} q(\delta) f(\delta)d\delta= \underline{q}.
\end{align*}

Next, we solve $\mathrm{(OP-B)}$ from an optimal control perspective again, and the results are summarized in Theorem \ref{thm2}. For clarity, we denote by $q^b(\delta),\ V^b({\delta})$ the optimal solutions to $\mathrm{(OP-B)}$. Further analysis indicates that $V^b({\delta})= 0,\ \forall \delta$, and the pricing scheme is charaterized by $p^b(\delta) = \Phi(\delta,q^b(\delta))$. By regarding $q(\delta)$ as a control variable, i.e., $q(\delta)=u(\delta)$, and introducing a state $\dot{x}(\delta) = u_1(\delta)f(\delta)$ with boundary constraints $x(\bar{\delta}) = \underline{q}$ and $x(\underline{\delta})=0$, we can reformulate $\mathrm{(OP-B)}$ as 
\begin{align*}
\mathrm{(OP-B'):}\
 \max_{\{u(\delta)\}}\ & \int_{\underline{\delta}}^{\bar{\delta}} \Big (\Phi(\delta,u(\delta)) -C(u(\delta))\Big) f(\delta)d\delta\\
\mathrm{s.t.}\ 
&\dot{x}(\delta) = u(\delta)f(\delta),\ x(\bar{\delta}) = \underline{q},\ x(\underline{\delta})=0.
\end{align*}
Note that $\mathrm{(OP-B')}$ is an optimal control problem with fixed initial and terminal state constraints.
The Hamiltonian of $\mathrm{(OP-B')}$ is 
\begin{equation}\label{Hamilton_2}
\begin{split}
H({x}(\delta), u(\delta), {\lambda}(\delta),\delta) = \big[\Phi(\delta,u(\delta)) 
-C(u(\delta))\big] \\ 
\cdot f(\delta) +\lambda(\delta) u(\delta)f(\delta),
\end{split}
\end{equation}
where $\lambda$ is the costate variable associated with the state dynamics. The maximum principle yields the following Hamilton system:
$
H({x}^b(\delta), u^b(\delta), {\lambda}^b(\delta),\delta) \geq H({x}^b(\delta), u(\delta), {\lambda}^b(\delta),\delta),
\dot{x}^b = u^b(\delta)f(\delta),
\dot{\lambda}^b  = 0,
\lambda(\bar{\delta})=\beta,
$
where $\beta$ is a real constant.

The first-order condition of \eqref{Hamilton_2} with respect to $u$ is
$
\frac{\partial H({x}^b(\delta), u(\delta), {\lambda}^b(\delta),\delta)}{\partial u}  
= (\frac{\partial \Phi(\delta,u(\delta))}{\partial u}  - \frac{dC(u(\delta))}{du})f(\delta)
+\lambda^b(\delta)f(\delta)= 0.
$
Further, the second-order conditions of Hamiltonian \eqref{Hamilton_2} with respective to $x$ and $u$ are nonpositive, and hence the obtained $u^b(\delta)$ is optimal.
Then, the optimal control $u^b(\delta)$ satisfies
$
(\delta-a\sigma e^{au^b(\delta)}+\beta) f(\delta)=  0,
$
which further yields
$
u^b(\delta) = \frac{1}{a}\ln \frac{\delta+\beta}{a\sigma}.
$
The constant $\beta$ can be solved from $\int_{\underline{\delta}}^{\bar{\delta}} u^b(\delta) f(\delta)d\delta= \underline{q}.$

We summarize the optimal contract for SaaS under the complete information in the following theorem.
\begin{theorem}\label{thm2}
When the SP has the complete incentive information of the IoT users, the optimal contracts $\{q^b(\delta),p^b(\delta)\}$ are designed as follows:
\begin{equation}
\begin{split}
q^b(\delta) &= \frac{1}{a}\ln \left(\frac{\delta+\beta}{a\sigma} \right),\\
p^b(\delta) &= \Phi(\delta,	q^b(\delta)) = \delta q^b(\delta),
\end{split}
\end{equation}
where $\beta$ is determined from $\int_{\underline{\delta}}^{\bar{\delta}} \ln \frac{\delta+\beta}{a\sigma} f(\delta)d\delta= a\underline{q}$.
\end{theorem}

\textit{Remark:} Theorem \ref{thm2} helps to identify the fundamental differences of optimal contracts designed under complete and incomplete information structures. Comparing with the designed optimal contracts $\{q^*(\delta), p^*(\delta)\}$ in Theorem \ref{thm1}, the sensing QoS mapping $q^b(\delta)$ and pricing function $p^b(\delta)$ in Theorem \ref{thm2} do not contain terms $\frac{F(\delta)-1}{f(\delta)}\leq 0$ and $\phi(\delta)\geq 0$, respectively. The different values of $\beta$ in $q^*(\delta)$ and $q^b(\delta)$ prohibit the conclusion that $q^*(\delta)\leq q^b(\delta)$ and $p^*(\delta)\leq p^b(\delta)$. Note that the constraint $\int_{\underline{\delta}}^{\bar{\delta}} q(\delta) f(\delta)d\delta= \underline{q}$ indicates the same mean QoS in two scenarios without/with asymmetric information between SP and users. Therefore, when $q^*(\delta)\neq q^b(\delta)$, $\forall \delta\in[\underline{\delta},\bar{\delta}]$, we can conclude that there exists at least one $\tilde{\delta}$ where $q^*(\tilde{\delta})=q^b(\tilde{\delta})$, and the IoT users in the benchmark case pay more for the service due to $\phi(\delta)\geq 0$, i.e., $p^*(\tilde{\delta})<p^b(\tilde{\delta})$. 
Another remark is that the total profit of sensing SP by providing the optimal contracts resulting from $\mathrm{(OP-B)}$ is no less than the one from $\mathrm{(OP)}$ due to the removal of IC constraint which enlarges the feasible decision space. The profit difference can be interpreted as the private user's type information cost which we will quantify in Section \ref{examples}.

\section{Optimal Contracts for General User's Type Distributions}\label{extension}
In this section, we investigate the scenarios when the density condition in Lemma \ref{Lemma2} does not hold. We provide an alternative maximum principle and a full characterization of optimal contracts for SaaS in this general case. 

\subsection{Maximum Principle and Optimality Analysis}

Following the notations in $(\mathrm{OP}'')$ except replacing $u$ with $x_3$ and introducing a new control variable $\mu$, we formulate the following problem:
\begin{align*}
\mathrm{(OP-E):}\\
\max_{\substack{\{\mu(\delta),x_1(\delta),\\ x_2(\delta),x_3(\delta)\}}}\ & \int_{\underline{\delta}}^{\bar{\delta}} \big (\Phi(\delta,x_3(\delta)) - x_1(\delta) -C(x_3(\delta))\big) f(\delta)d\delta\\
\mathrm{s.t.}\quad & \dot{x}_1(\delta) = x_3(\delta),\ x_1(\underline{\delta})= 0,\\
&\dot{x}_2(\delta)=x_3(\delta)f(\delta),\ x_2(\bar{\delta}) = \underline{q},\ x_2(\underline{\delta})=0,\\
&\dot{x}_3(\delta) = \mu(\delta),\ \mu(\delta)\geq 0.
\end{align*}

Note that $\mathrm{(OP-E)}$ is an optimal control problem with three state variables $x_1,\ x_2,\ x_3$ and a control variable $\mu$, where the initial points of $x_1$ and $x_2$, and the boundary points of $x_2$ are fixed.

The Hamiltonian of $\mathrm{(OP-E)}$ can be written as
$
H(\mathbf{x}(\delta), \mu(\delta), \boldsymbol{\lambda}(\delta),\delta) = [\Phi(\delta,x_3(\delta)) - x_1(\delta)-C(x_3(\delta))] \cdot f(\delta) + \lambda_1(\delta) x_3(\delta)  +\lambda_2(\delta) x_3(\delta)f(\delta)+\lambda_3(\delta) \mu(\delta),
$
where $\mathbf{x}=[x_1,x_2, x_3]^T$ and $\boldsymbol{\lambda}= [\lambda_1,\lambda_2, \lambda_3]^T$.
To differentiate with the optimal solution $(\mathbf{x}^*(\delta),u^*(\delta))$ in Theorem \ref{thm1}, we denote by $(\mathbf{x}^o(\delta),\mu^o(\delta))$ the optimal solution to the cases with general user's type distribution. Using the Pontryagin maximum principle, we obtain $(\mathbf{x}^o(\delta),\mu^o(\delta))$ by solving the  Hamilton system:
\begin{align}
H&(\mathbf{x}^o(\delta), \mu^o(\delta), \boldsymbol{\lambda}^o(\delta),\delta) \geq H(\mathbf{x}^o(\delta), \mu(\delta), \boldsymbol{\lambda}^o(\delta),\delta),\label{H_1_E}\\
\dot{x}_1^o &= \frac{\partial H(\mathbf{x}^o(\delta), \mu^o(\delta), \boldsymbol{\lambda}^o(\delta),\delta)}{\partial \lambda_1(\delta)} = x_3^o(\delta),\label{H_2_E}\\
\dot{x}_2^o &= \frac{\partial H(\mathbf{x}^o(\delta), \mu^o(\delta), \boldsymbol{\lambda}^o(\delta),\delta)}{\partial \lambda_2(\delta)} = x_3^o(\delta)f(\delta),\label{H_3_E}\\
\dot{x}_3^o &= \frac{\partial H(\mathbf{x}^o(\delta), \mu^o(\delta), \boldsymbol{\lambda}^o(\delta),\delta)}{\partial \lambda_3(\delta)} = \mu^o(\delta),\label{H_4_E}\\
\dot{\lambda}_1^o &= -\frac{\partial H(\mathbf{x}^o(\delta), \mu^o(\delta), \boldsymbol{\lambda}^o(\delta),\delta)}{\partial x_1(\delta)} = f(\delta),\label{H_5_E}\\
\dot{\lambda}_2^o &= -\frac{\partial H(\mathbf{x}^o(\delta), \mu^o(\delta), \boldsymbol{\lambda}^o(\delta),\delta)}{\partial x_2(\delta)} = 0,\label{H_6_E}\\
\dot{\lambda}_3^o &= -\frac{\partial H(\mathbf{x}^o(\delta), \mu^o(\delta), \boldsymbol{\lambda}^o(\delta),\delta)}{\partial x_3(\delta)}\notag \\
&= -\left(\frac{\partial \Phi(\delta,x_3(\delta))}{\partial x_3(\delta)} -\frac{dC(x_3(\delta))}{dx_3(\delta)}\right)f(\delta)\notag\\
&\quad -\lambda_1^o(\delta)-\lambda_2^o(\delta)f(\delta),\label{H_7_E}\\
&\lambda_1(\bar{\delta})  =0, \label{H_8_E}\\
 &\lambda_2(\bar{\delta})\ \mathrm{is\ a\ constant}, \label{H_9_E}\\
 &\lambda_3(\underline{\delta})= \lambda_3(\bar{\delta}) = 0.\label{H_10_E}
\end{align}
Note that \eqref{H_8_E} and \eqref{H_9_E} are boundary conditions which are similar to the ones in \eqref{H_6} and \eqref{H_7}. In $(\mathrm{OP-E})$, we include another state variable $x_3$ which does not have initial and terminal constraints. Then, based on the maximum principle \cite{kirk2012optimal}, the corresponding costate variable $\lambda_3$ at time $\underline{\delta}$ and $\bar{\delta}$ should equal to the derivative of the initial and terminal payoff with respect to the state $x_3$, respectively. In $(\mathrm{OP-E})$, the objective function does not contain individual initial and terminal utilities, and thus we obtain condition \eqref{H_10_E}. 

First, similar to \eqref{lamb1} and \eqref{lamb2}, we observe that
\begin{align}
{\lambda}_1^o(\delta) &= F(\delta)-1,\label{lamb1_E}\\
{\lambda}_2^o(\delta) &= \beta,\label{lamb2_E}
\end{align}
where the constant $\beta$ can  be determined using\eqref{reputation_2} after the QoS mapping $q^o(\delta)$ is characterized.

In addition, by integrating \eqref{H_7_E}, we obtain
\begin{equation}\label{lambda_3}
\begin{split}
\lambda_3^o(\delta) = - \int_{\underline{\delta}}^{\delta} \left(\frac{\partial \Phi(\delta,x_3(\delta))}{\partial x_3(\delta)} -\frac{dC(x_3(\delta))}{dx_3(\delta)}\right)f(\delta)\\
+\lambda_1^o(\delta)+\lambda_2^o(\delta)f(\delta)d\delta.
\end{split}
\end{equation}
Using the transversality conditions $\lambda_3(\underline{\delta})= \lambda_3(\bar{\delta}) = 0$ yields
$
\lambda_3(\bar{\delta}) = - \int_{\underline{\delta}}^{\bar{\delta}} (\frac{\partial \Phi(\delta,x_3(\delta))}{\partial x_3(\delta)} -\frac{dC(x_3(\delta))}{dx_3(\delta)})f(\delta)\
+\lambda_1^o(\delta)+\lambda_2^o(\delta)f(\delta)d\delta =0.
$
Furthermore, \eqref{H_1_E} indicates that $\mu^o(\delta)$ maximizes $H$ with $\mu^o(\delta)\geq 0$. Note that in the Hamiltonian $H$, the last term $\lambda_3(\delta)\mu(\delta)$ imposes a non-positive value constraint on $\lambda_3(\delta)$. Otherwise, $H$ is unbounded from above due to $\mu(\delta)\geq 0$. Then, to ensure the feasibility of maximization, we have $\lambda_3(\delta)\leq 0$ which is equivalent to 
$
 \int_{\underline{\delta}}^{\delta} (\frac{\partial \Phi(\delta,x_3(\delta))}{\partial x_3(\delta)} -C'(x_3(\delta)))f(\delta)\
+\lambda_1^o(\delta)+\lambda_2^o(\delta)f(\delta)d\delta\geq 0.
$
Thus, when $\lambda_3(\delta)<0$, $\dot{x}_3^o(\delta) = \mu^o(\delta)=0.$ Therefore, the complementary slackness condition can be written as follows, $\forall \delta\in[\underline{\delta},\bar{\delta}]$,
\begin{equation}\label{C_S}
\begin{split}
 \dot{x}_3^o(\delta) \int_{\underline{\delta}}^{\delta} \left(\frac{\partial \Phi(\delta,x_3^o(\delta))}{\partial x_3^o(\delta)} -\frac{dC(x_3(\delta))}{dx_3(\delta)}\right)f(\delta)\\
+\lambda_1^o(\delta)+\lambda_2^o(\delta)f(\delta)d\delta = 0.
\end{split}
\end{equation}

We can verify that the maximum principle \eqref{H_1_E}--\eqref{H_10_E} is also sufficient for optimality as the associated Hamiltonian equation is concave in both $\mathbf{x}$ and  $\mu$. Furthermore, the Hamiltonian is strictly concave in $x_3$ and other states are uniquely determined by $x_3$. Thus, the optimal control and optimal state trajectory are unique  \cite{mangasarian1966sufficient}. We next explicitly characterize this optimal solution.

\subsection{Characterization of Optimal Contracts}
We next analyze the optimal contracts in two regimes regarding $\dot{x}_3^o(\delta)$, i.e., $\dot{x}_3^o(\delta)>0$ and $\dot{x}_3^o(\delta)=0$. Based on \eqref{C_S}, in the interval of $\delta$ that $\dot{x}_3^o(\delta)>0$, then $\lambda_3^o(\delta) = 0$ for all $\delta$ in this interval, which further indicates $\dot{\lambda}_3^o=0$. Hence, from \eqref{H_7_E}, the following equation holds:
$
(\frac{\partial \Phi(\delta,x_3(\delta))}{\partial x_3(\delta)} -\frac{dC(x_3(\delta))}{dx_3(\delta)})f(\delta)+\lambda_1^o(\delta)
+\lambda_2^o(\delta)f(\delta) =0,
$
which is exactly the same maximality condition presented in \eqref{maximality}, where $x_3(\delta)$ plays the role as $u({\delta})$. Following the same analysis in Section \ref{optimal_design_results}, the optimal solutions to $x_1^o$, $x_2^o$, $x_3^o$, $\lambda_1^o$ and $\lambda_2^o$ in Hamilton system \eqref{H_1_E}--\eqref{H_10_E} coincide with $x_1^*$, $x_2^*$, $u^*$, $\lambda_1^*$ and $\lambda_2^*$ in Hamilton system \eqref{H_1}--\eqref{H_7}.
Thus, we can conclude that if $x_3^o(\delta)$ is strictly increasing over some interval and recall the notation $x_3(\delta) = q(\delta)$, the solution $q^o(\delta)$ in this section should be the same as the one $q^*(\delta)$ in Theorem \ref{thm1}. 

\begin{figure*}[t]
  \centering  
  \subfigure[Case I: $\delta_1=\underline{\delta}$]{%
    \includegraphics[width=0.3\textwidth]{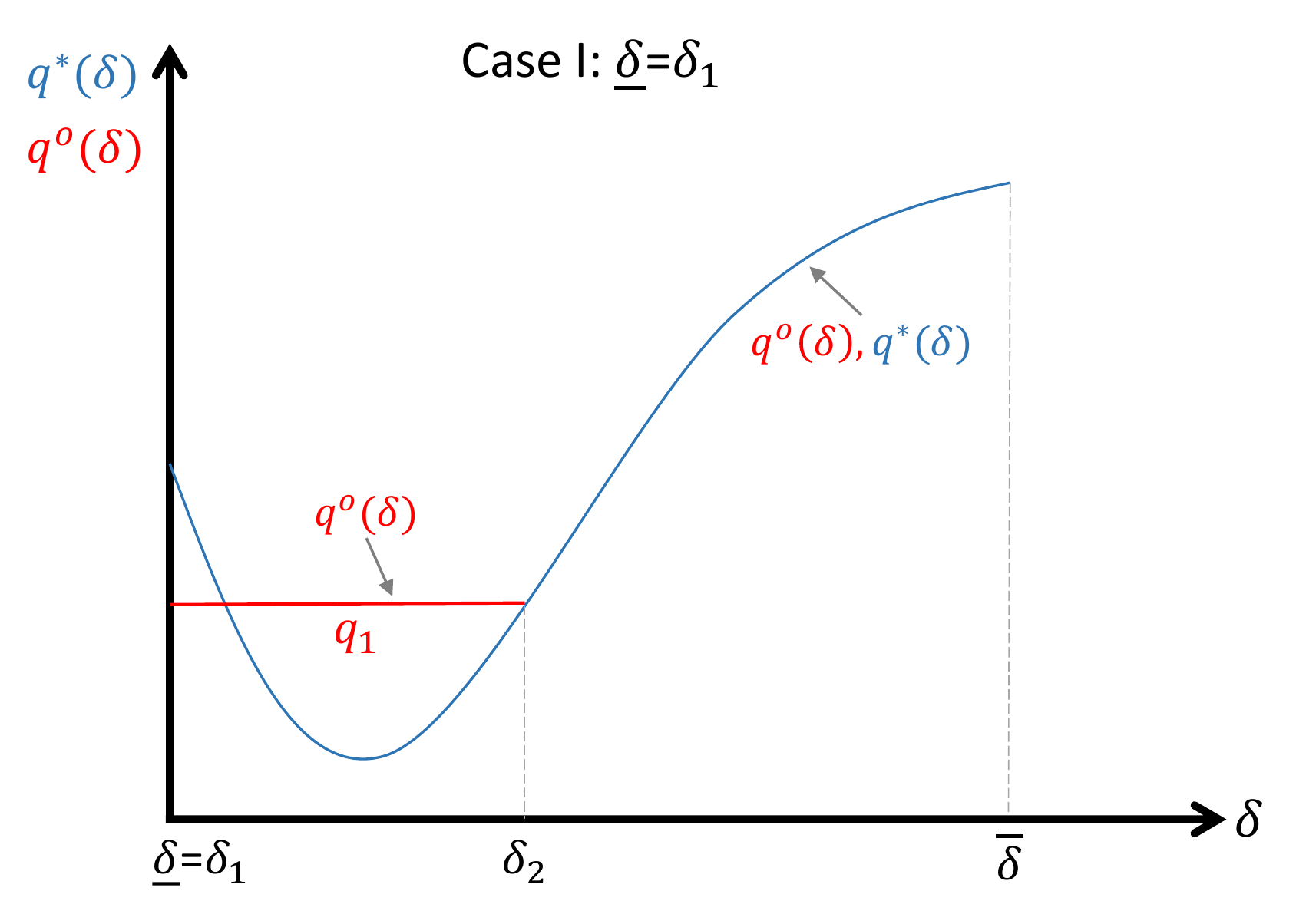}%
    \label{case1}%
  }%
  \hfill
  \subfigure[Case II: $\underline{\delta}<\delta_1<\delta_2<\bar{\delta}$]{%
    \includegraphics[width=0.3\textwidth]{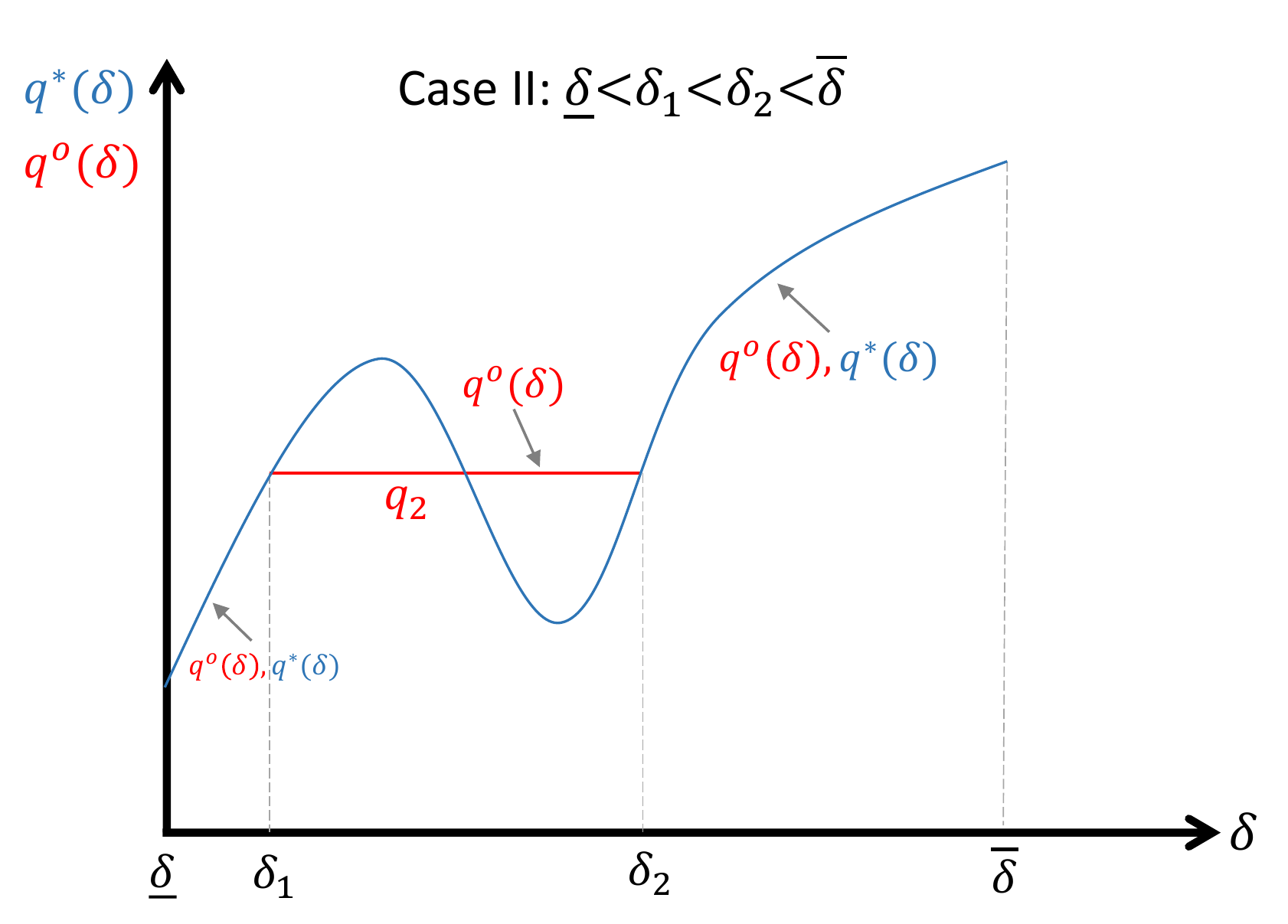}%
    \label{case2}%
  }%
  \hfill
  \subfigure[Case III: $\delta_2=\bar{\delta}$]{%
    \includegraphics[width=0.3\textwidth]{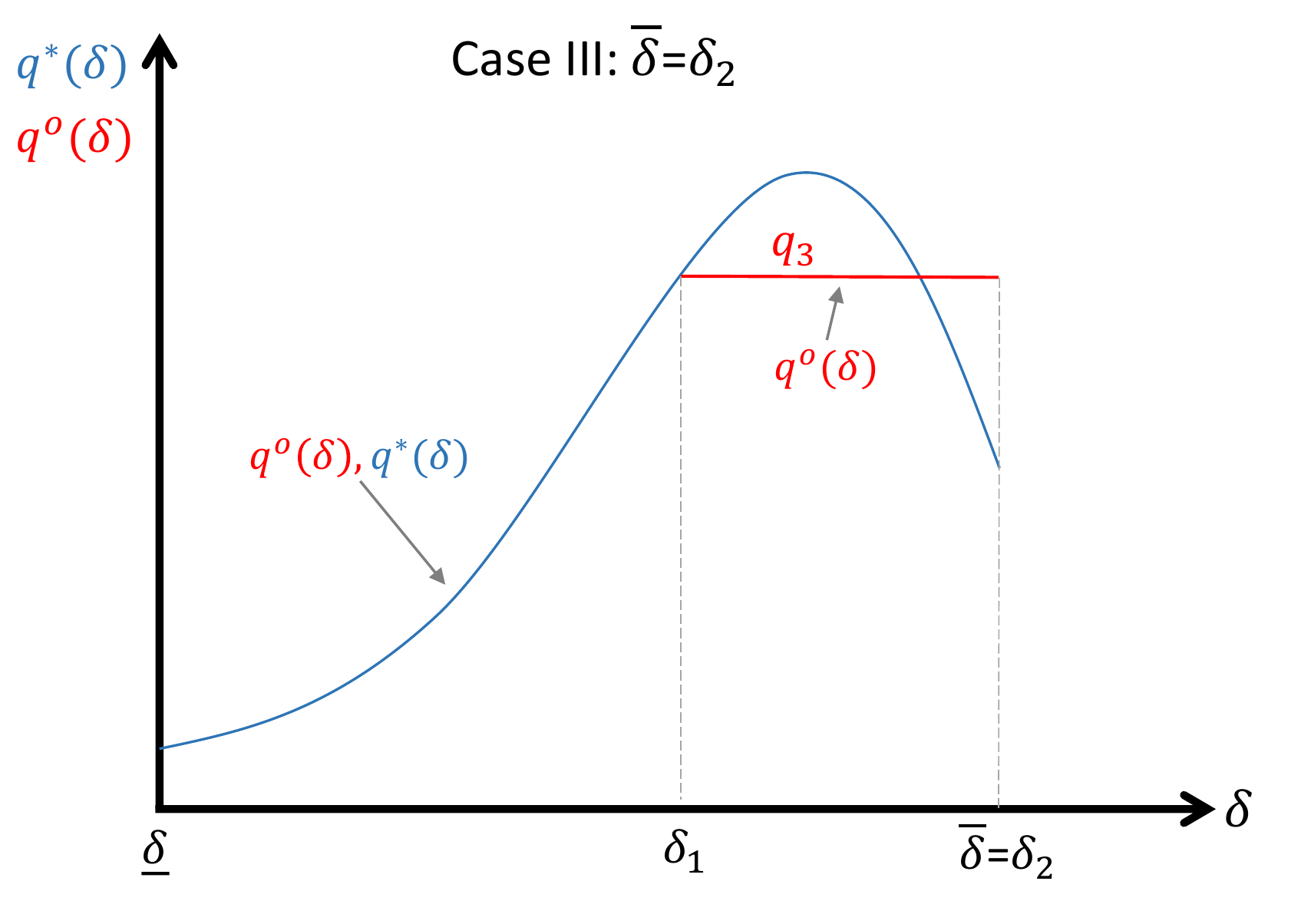}%
    \label{case3}%
  }%
  \caption{In all three figures, $q^o(\delta)$ and $q^*(\delta)$ represent the QoS of SaaS with and without considering the monotonicity constraint, respectively. In addition, the optimal solution $q^o(\delta)$ coincides with $q^*(\delta)$ over some interval except $\delta\in[\delta_1,\delta_2]$ in three cases. For  $\delta\in[\delta_1,\delta_2]$, $q^o(\delta)$ is nondiscriminative and admits constant values $q_1$  $q_2$ and $q_3$ in (a), (b) and (c), respectively.}
  \label{cases}
\end{figure*}

In the other regime of $\dot{x}_3^o(\delta)=0$, ${x}_3^o(\delta)$ is unchanged. Then, the remaining task is to determine the intervals of $\delta$ in which $q^o(\delta)$ admits a constant, and hence the service price is nondiscriminative. Note that these intervals definitely include the ones when $q^*(\delta)$ is decreasing, i.e., the monotonicity constraint of sensing QoS is violated. For notational convenience, let $[\delta_1,\delta_2]$ be the interval when $q^o(\delta)$ is a constant, $\delta\in [\delta_1,\delta_2]$. We know that for $\delta<\delta_1$ and $\delta>\delta_2$, $q^o(\delta)$ is increasing, and thus  $\dot{x}_3^o(\delta)>0$. Based on \eqref{C_S}, we obtain condition $\lambda_3^o(\delta) = 0$. Since the costate variable $\lambda_3^o$ is continuous, then at the critical  points $\delta_1$ and $\delta_2$, 
$\lambda_3^o(\delta_1)=\lambda_3^o(\delta_2) = 0,$
and using \eqref{lambda_3} yields
\begin{equation}\label{flat_interval}
\begin{split}
\int_{\delta_1}^{\delta_2} \left(\frac{\partial \Phi(\delta,q(\delta))}{\partial q(\delta)} -\frac{dC(q(\delta))}{dq(\delta)}\right)f(\delta)\\
+\lambda_1^o(\delta)+\lambda_2^o(\delta)f(\delta)d\delta=0.\end{split}
\end{equation}

To this end, we discuss three possible cases that $q^o(\delta)$ is nondiscriminative over $\delta\in[\delta_1,\delta_2]$ subsequently. When analyzing $q^o(\delta)$, we constantly refer to the optimal solution $q^*(\delta)$ in Theorem \ref{thm1}. Besides, we assume that both $\lambda_1^o$ and $\lambda_2^o$ are known through \eqref{lamb1_E} and \eqref{lamb2_E} with an exception of $\beta$ to be specified later.

\textit{Case I: ($\delta_1=\underline{\delta}$).} In this case, \eqref{flat_interval} is reduced to 
\begin{equation}\label{flat_interval_c1}
\begin{split}
\int_{\underline{\delta}}^{\delta_2} \left(\frac{\partial \Phi(\delta,q_1)}{\partial q(\delta)} -\frac{dC(q_1)}{dq(\delta)}\right)f(\delta)\\
+\lambda_1^o(\delta)+\lambda_2^o(\delta)f(\delta)d\delta=0,\\
q_1 = q^*(\delta_2).
\end{split}
\end{equation}
One illustrative example for this scenario is shown in Fig. \ref{case1}, where for $\delta\in[\delta_2,\bar{\delta}]$, $q^o(\delta) = q^*(\delta)$.
In addition, the constant value $q_1$ is no greater than $q^*(\underline{\delta})$, i.e., $q_1\leq q^*(\underline{\delta})$. We prove this result by contradiction. If $q_1>q^*(\underline{\delta})$, then $q_1>q^*(\tilde{\delta})$ for any $\tilde{\delta}$ close enough to $\underline{\delta}$. Along with the entire trajectory $q^*(\delta)$, we introduce a virtual variable $\lambda_3^*(\delta)$ which is a counterpart of $\lambda_3^o(\delta)$, and thus we have $\lambda_3^*(\delta)=0$. Recall the notation $x_3=q$, and then the partial integrand $\frac{\partial \Phi(\delta,q(\delta))}{\partial q(\delta)} -\frac{dC(q(\delta))}{dq(\delta)}$ in \eqref{lambda_3} decreases when the value of $q$ increases due to the convexity of cost function $C$. Thus, the entire $\lambda_3^o(\delta)$ increases if $q$ becomes larger. Therefore, for $\tilde{\delta}$ close enough to $\underline{\delta}$ and based on the assumption $q_1>q^*(\underline{\delta})$, we obtain $\lambda_3^o(\tilde{\delta})>\lambda_3^*(\tilde{\delta})=0$, contradicting the condition $\lambda_3^o(\delta)\leq 0,\ \forall \delta\in[\underline{\delta},\bar{\delta}]$.
Therefore, we can obtain $\delta_2$ and the corresponding value $q_1$ by solving two equations in \eqref{flat_interval_c1}.

\textit{Case II: ($\underline{\delta}<\delta_1<\delta_2<\bar{\delta}$).} When the interval $[\delta_1,\delta_2]$ lies in the interior of the entire regime $\delta$, \eqref{flat_interval} becomes
\begin{equation}\label{flat_interval_c2}
\begin{split}
\int_{\delta_1}^{\delta_2} \left(\frac{\partial \Phi(\delta,q_2)}{\partial q(\delta)} -\frac{dC(q_2)}{dq(\delta)}\right)f(\delta)\\
+\lambda_1^o(\delta)+\lambda_2^o(\delta)f(\delta)d\delta=0,\\
q_2 = q^*(\delta_1) = q^*(\delta_2).
\end{split}
\end{equation}
We can solve for two unknowns $\delta_1$ and $\delta_2$ based on \eqref{flat_interval_c2}, and subsequently we obtain $q_2$. Case II is depicted in Fig. \ref{case2}.

\textit{Case III: ($\delta_2=\bar{\delta}$).}  When $\delta_2$ coincides with the end-point $\bar{\delta}$, \eqref{flat_interval} can be written as
\begin{equation}\label{flat_interval_c3}
\begin{split}
\int_{\delta_1}^{\bar{\delta}} \left(\frac{\partial \Phi(\delta,q_3)}{\partial q(\delta)} -\frac{dC(q_3)}{dq(\delta)}\right)f(\delta)\\
+\lambda_1^o(\delta)+\lambda_2^o(\delta)f(\delta)d\delta=0,\\
q_3 = q^*(\delta_1).
\end{split}
\end{equation}
Fig. \ref{case3} presents an example of case III. Similar to the analysis in Case I, the value of $q_3$ satisfies $q_3\geq q^*(\bar{\delta})$. Furthermore, $\delta_1$ and $q_3$ can be obtained by solving \eqref{flat_interval_c3}.

Note that in the optimal contracts, the intervals over which $q^o(\delta)$ admitting a constant value can be a combination of the three cases, and there could exist multiple interior intervals as the one shown in Fig. \ref{case2}. 
Another essential point is to determine $\lambda_2^o=\beta$ in \eqref{flat_interval_c1}--\eqref{flat_interval_c3}. As the analysis in Section \ref{optimal_design_results}, the unknown constant $\beta$ can be derived using the constraint \eqref{reputation_2}. However, \eqref{reputation_2} needs a full expression of optimal $q^o$ beforehand. Therefore, two procedures including the derivation of optimal solution $q^o$ from \eqref{flat_interval_c1}--\eqref{flat_interval_c3} and the obtaining $\lambda_2(\delta) = \beta$ by \eqref{reputation_2} are intertwined. To design the optimal $q^o(\delta)$, we thus should solve the equations \eqref{flat_interval_c1}--\eqref{flat_interval_c3} together with \eqref{reputation_2} in a holistic manner.  With derived $q^o(\delta)$, the service pricing function $p^o(\delta)$ then can be characterized with similar steps in Section \ref{optimal_design_results}.

We summarize the optimal contracts for SaaS under general user's type distribution in the following theorem.

\begin{theorem}\label{thm3}
For a general user's type distribution $f(\delta)$ where $2f^2(\delta)+(1-F(\delta))f'(\delta)>0$ does not hold, the optimal contracts $\{q^o(\delta),p^o(\delta)\}$ designed by the SP are detailed as follows. The QoS mapping $q^o(\delta)$ is piecewise continuous and weakly increasing over $\delta\in[\underline{\delta},\bar{\delta}]$.
\begin{enumerate}
\item $q^o(\delta)$ and $p^o(\delta)$ coincide with $q^*(\delta)$ and $p^*(\delta)$ in Theorem \ref{thm1} except on a finite number $N$ of disjoint intervals $I_n = (\delta_1^n,\delta_2^n)$, for $n=1,...,N$, and $\delta_1^n$ and $\delta_2^n$ increase with $n$. Furthermore,, $q^o(\delta)=q_n$, $\forall \delta\in I_n$. 
\item For the interior interval $I_n$ where $\delta_1^n\neq \underline{\delta}$ and  $\delta_2^n\neq \bar{\delta}$, the optimal $q^o(\delta)$ satisfies
\begin{equation}\label{thm3_E_1}
\begin{split}
\int_{\delta_1^n}^{\delta_2^n} \left(\frac{\partial \Phi(\delta,q_n)}{\partial q} -\frac{dC(q_n)}{dq}\right)f(\delta)\\
+\lambda_1^o(\delta)+\lambda_2^o(\delta)f(\delta)d\delta=0,\\
q_n = q^*(\delta_1^n) = q^*(\delta_2^n).
\end{split}
\end{equation}
\item If $\delta_1^1=\underline{\delta}$, i.e., the interval $I_1$ starts with $\underline{\delta}$, then the optimal $q^o(\delta)$ satisfies
\begin{equation}\label{thm3_E_2}
\begin{split}
\int_{\underline{\delta}}^{\delta_2^1} \left(\frac{\partial \Phi(\delta,q_1)}{\partial q} -\frac{dC(q_1)}{dq}\right)f(\delta)\\
+\lambda_1^o(\delta)+\lambda_2^o(\delta)f(\delta)d\delta=0,\\
q_1 = q^*(\delta_2^1)\leq q^*(\underline{\delta}).
\end{split}
\end{equation}
\item If $\delta_2^N=\bar{\delta}$, i.e., the interval $I_N$ ends with $\bar{\delta}$, then the optimal $q^o(\delta)$ satisfies
\begin{equation}\label{thm3_E_3}
\begin{split}
\int_{\delta_1^N}^{\bar{\delta}} \left(\frac{\partial \Phi(\delta,q_N)}{\partial q} -\frac{dC(q_N)}{dq}\right)f(\delta)\\
+\lambda_1^o(\delta)+\lambda_2^o(\delta)f(\delta)d\delta=0,\\
q_N = q^*(\delta_1^N)\geq q^*(\bar{\delta}).
\end{split}
\end{equation}
\item Based on \eqref{thm3_E_1}--\eqref{thm3_E_3} and together with \eqref{reputation_2}, \eqref{lamb1_E}, \eqref{lamb2_E}, $q_n$, $\delta_1^n$ and $\delta_2^n$, $n=1,...,N$, can be computed. After obtaining the sensing QoS function $q^o(\delta)$, the optimal pricing $p^o(\delta)$ can be derived via the relation 
\begin{equation}\label{price_E}
p^o(\delta) = \Phi(\delta, q^o(\delta)) - \phi(\delta),
\end{equation}
where $\dot \phi(\delta) =  q^o(\delta)$ with $\phi(\underline{\delta})=0$.
\end{enumerate}
\end{theorem}

\textit{Remark:} For the intervals where $q^o(\delta) = q^*(\delta)$, $p^o(\delta)$ is monotonically increasing. For $\delta\in I_n,\ n=1...,N$, $q^o(\delta)$ is a constant and then $\dot{q}^o(\delta)=0$. Based on \eqref{price_E} and $\Phi(\delta,q^o(\delta)) = \delta q^o(\delta)$, we obtain $\dot{p}^o(\delta) = \delta \dot{q}^o(\delta)+ q^o(\delta)-\dot \phi(\delta) = 0$. Therefore, IoT users with a type lying in the same interval $I_n$, $n=1,...,N$, are provided with a menu of contracts with the same quality of sensing data as well as the service price.

\subsection{Some Analytical Results}
We end up this section by presenting analytical results on the pricing of sensing services. These results give insights on the obtained solutions, and they also contribute to the design of practical market-based contracts.

\textit{(1) Structure of the optimal contracts:} Comparing with the optimal contracts in Theorem \ref{thm1}, the ones  in Theorem \ref{thm3} have an additional feature of nondiscriminative service intervals. Specifically, in addition to the profit maximization and service reputation construction of SP, the IC constraints of users are completely considered in the contracts, where the additional monotonicity part is reflected by \eqref{thm3_E_1}--\eqref{thm3_E_3}. Note that the nondiscriminative pricing \textit{reduces the diversity} of service provisions to the IoT users which has an interpretation that the SP treats heterogeneous users equally.
Different with the contracts in Theorem \ref{thm1} of full separation, the pooling behavior (users of different types are offered with the same contract) in Theorem \ref{thm3} due to irregular type distribution is to ensure the incentive compatibility of designed optimal contracts.

\textit{(2) Number of intervals with nondiscriminative pricing:} Fig. \ref{cases} shows that the intervals with a decreasing $q^*(\delta)$ are included in $I_n$, $n=1,...,N$. Then, $N$ is equal to the number of peaks (local maximum) of $q^*(\delta)$. Based on Theorem \ref{thm1}, we analyze the monotonicity of $\frac{F(\delta)-1}{f(\delta)}+ \delta$, indicating that the number of nondiscriminative pricing regimes $N$ coincides with the number of intervals where $ 2f^2(\delta)+(1-F(\delta))f'(\delta)$ takes a negative value.

\textit{(3) Nondiscriminative pricing for all users:} When $q^*(\delta)$ is decreasing over $\delta\in [\underline{\delta},\bar{\delta}]$, then based on Theorem \ref{thm3}, the optimal service pricing $q^o(\delta)$ is nondiscriminative for all types of users. In this scenario,  we obtain $2f^2(\delta)+(1-F(\delta))f'(\delta)<0$ for all $\delta$. From Lemma \ref{Lemma2}, an equivalent condition is that $\frac{1-F(\delta)}{f(\delta)}$ increases over $\delta$. We summarize the results in the following lemma.

\begin{lemma}\label{Lemma3}
The optimal contracts $\{q^o(\delta),p^o(\delta)\}$ are nondiscriminative for all $\delta$ if $\frac{1-F(\delta)}{f(\delta)}$ increases over $\delta\in [\underline{\delta},\bar{\delta}]$. An alternative equivalent condition leading to the results is that function $\log[1-F(\delta)]$ is strictly convex.
\end{lemma}

Some typical distributions satisfying Lemma \ref{Lemma3} are worth highlighting. One example is when $f(\delta)$ is a gamma distribution for parameter $\alpha<1$, i.e., $f(\delta) = \frac{\psi^{\alpha}\delta^{\alpha-1}\exp(-\psi \delta)}{\Gamma(\alpha)}$, where $\delta\geq 0$ and $\Gamma(\delta)$ is a complete Gamma function. Another example is when $f(\delta)$ admits a Weibull distribution under $\alpha<1$, i.e., $f(\delta) = \psi \alpha \delta^{\alpha-1}\exp(-\psi \delta^\alpha)$, $\delta\geq 0$. In both types of distributions, most of the IoT users are with type $\underline{\delta}=0$ or close to $\underline{\delta}$, and its number decreases exponentially as the parameter $\delta$ increases. Therefore, the SP designs nondiscriminative contracts for all users, extracting the profits from the majority of customers in the market. Moreover, this nondiscriminative service provision mechanism aligns with the phenomenon of \textit{focusing on the majority}, where the small group of users with larger types are treated in a homogeneous manner as the major population nested in lower types.

\textit{(4) Invariant nondiscriminative service pricing:} One natural question is the impact of convexity of $\log[1-F(\delta)]$ on the service price. For various type distributions $f(\delta)$ satisfying the condition in Lemma \ref{Lemma3}, we show that the convexity of $F(\delta)$ has no influence on the neutral service pricing. Specifically, based on the constraint $\int_{\underline{\delta}}^{\bar{\delta}} q^o(\delta) f(\delta)d\delta= \underline{q}$, where $q^o(\delta)=q^c,\ \forall \delta$, we obtain $q^c \int_{\underline{\delta}}^{\bar{\delta}} f(\delta)d\delta= \underline{q}$. Therefore, under the the nondiscriminative pricing of sensing services, the QoS is $q^c = \underline{q}$ for all users. Furthermore, the IR constraint $V(\underline{\delta})=0$ leads to the optimal constant pricing $p^c=\underline{\delta}\underline{q}$. Hence, whenever the SP offers a nondiscriminative price scheme to all IoT users, the price must be invariant equaling to $\underline{\delta}\underline{q}$ in spite of the user's type distributions.

\section{Case Studies: UAV-Enabled Virtual Reality}\label{examples}
In this section, we apply the SaaS paradigm to UAV-enabled virtual reality as depicted in Fig. \ref{VR} to illustrate the optimal contract design principles. We envision a large VR service market in the future, and thus a huge number of users will purchase the VR services. This SaaS paradigm can be also applied to other personalized data related service provision scenarios, such as virtual tourism. This virtual service modality becomes popular under the current disruptions caused by COVID-19 pandemic worldwide.

\subsection{UAV-Enabled VR Setting}\label{VR_case_setting}
The VR quality can be quantified by user experience related metrics, including the resolution of the captured scene of UAV ($\tilde{q}_1$), the delay in sensing data transmission ($\tilde{q}_2$), and the reliability of UAV communicating with the tower ($\tilde{q}_3$). Specifically, for the resolution quality $\tilde{q}_1$, it can be in the general classes of 240p, 360p, 480p, 720p, 1080p (commonly available options such as in the streaming services), and the qualities between these classes. The delay $\tilde{q}_2$ is composed of factors including processing delay, queuing delay, transmission delay, and propagation delay of sensing data. The delay can be reduced by using a dedicated network that streamlines the network path, which is more costly for the sensing service provider. The tolerable end-to-end delay of modern VR applications is of an order of milliseconds, and a desired QoS has it less than 1 or 2 milliseconds \cite{Tactile_Internet}. The communication reliability $\tilde{q}_3$ between UAV and tower can be measured by the success rate that data packets are transmitted. According to a video QoS tutorial by Cisco \cite{video_qos}, the reliability should be above 99\% for a high QoS, and it is between 99.5\% and 95\% depending on the specific type of services. The reliability above is quantified by the packet loss rate.

We can aggregate these major metrics into a single measure $q$ taking values in the real space. More specifically, the QoS $q$ can be determined by a linear combination in a form of $\kappa_1 \tilde{q}_1+\kappa_2 \tilde{q}_2+\kappa_3 \tilde{q}_3$, where $\kappa_i$, $i=1,2,3$, are positive weighting factors. Equal weighting refers to the scenario with $\kappa_1=\kappa_2=\kappa_3=1/3$. To differentiate the delivered services and pricing in terms of metrics considered, we consider that, comparing with a small $q$, a larger $q$ has all higher values in $\tilde{q}_1$, $\tilde{q}_2$, and $\tilde{q}_3$.  This modeling also fits the real-world scenario well, as the customers choose a higher QoS should receive better service in every factor considered (resolution, delay, reliability) by paying more service fee.  We anticipate a large VR service market in the future, and thus a huge number of users will purchase the VR services.
We further specify the mean QoS $\underline{q} = 5$. As the sensing QoS is a mapping considering various metrics, we set the mean QoS $\underline{q} = 5$ corresponding to the service with 720p resolution, 0.15sec delay, and  97$\%$ UAV transmission reliability.  After obtaining the QoS in the optimal contract later on, we can reversely map $q$ to the three specific metrics considered. Based on the current technologies in communication and VR, we consider the resolution, delay, and reliability admit a value from 240p to 1080p, 0.5 ms to 5 ms, and 0.95$\%$ to 0.99$\%$, respectively. Note that in the optimal mechanism design, higher types of users receive better quality of VR service from the SP.

As depicted in Fig. \ref{VR_real_data}, the user's type distribution admits $f(\delta) = 0.952e^{-0.952\delta}$, and thus $F(\delta) =  1-e^{-0.952\delta}$. These distribution functions are aligned with the market data as discussed in Example 1 in Section \ref{type_data_quality}.

\subsection{Optimal Contracts under Hidden Information}
Based on Corollary \ref{corollary_q_expnential}, we depict the optimal contracts of VR services in Fig. \ref{example} with various values of $a$. The weighting factor $\sigma$ admits a value of 0.16, which gives a reasonable comparison between the service charging fee and the cost of providing the service. In the cases with parameter $a=0.47,0.49, 0.51$, and using the results in Section \ref{sec:exp}, we obtain $\beta = 1.14, 1.215, 1.315$, respectively. With these selected parameters, the obtained service pricing also matches with the data market. One observation is that both the VR pricing and the QoS mappings are monotonically increasing with the user's type, leading to an incentive compatible contract. Another phenomenon is that as $a$ increases, the VR QoS is decreasing for a given user's type under the regime $\delta>0.47$ as shown in Fig. \ref{q}. The reason is that a larger $a$ indicates a higher service cost of the SP which leads to a degraded VR QoS. Thus, the VR pricing decreases as well for a given $\delta$ as illustrated in Fig. \ref{p}. Different with the findings in regime $\delta>0.47$, the VR QoS increases with the parameter $a$ when $\delta<0.47$, showing that a larger cost of the SP provides a better VR service for the customers of type $\delta<0.47$ while the customers paying less. Note that the mean VR QoS $\underline{q}$ stays the same for all investigated cases. Then, to maintain a constant reputation that the VR SP builds in the market, the received QoS for customers of type $\delta<0.47$ should increase with $a$ comparing with those of $\delta>0.47$. This phenomenon also aligns with the fact that at the early stage of VR services promotion ($a$ is large), the SP focuses more on the types of customers with a large population in the market (small $\delta$ in the exponential distribution), by providing a relatively better VR service. Based on the VR application modeling in Section \ref{VR_case_setting}, Fig. \ref{q_metrics} presents the specific sensing QoS in terms of the considered resolution, delay, and reliability metrics. Under the the designed optimal contracts $\{p^*(\delta),q^*(\delta)\}$, Fig. \ref{case1_utility_SP} shows the corresponding utility of SP. As $a$ increases which yields a larger service cost, the SP's aggregate revenue decreases accordingly. In addition, for some small types $\delta$ close to $\underline{\delta}$, $U(\delta)$ can be negative. This phenomenon indicates that the SP makes most of the profits from the users who demand a high VR QoS.
 
\begin{figure}[t]
  \centering
  \subfigure[VR pricing $p^*(\delta)$]{
    \includegraphics[width=0.48\columnwidth]{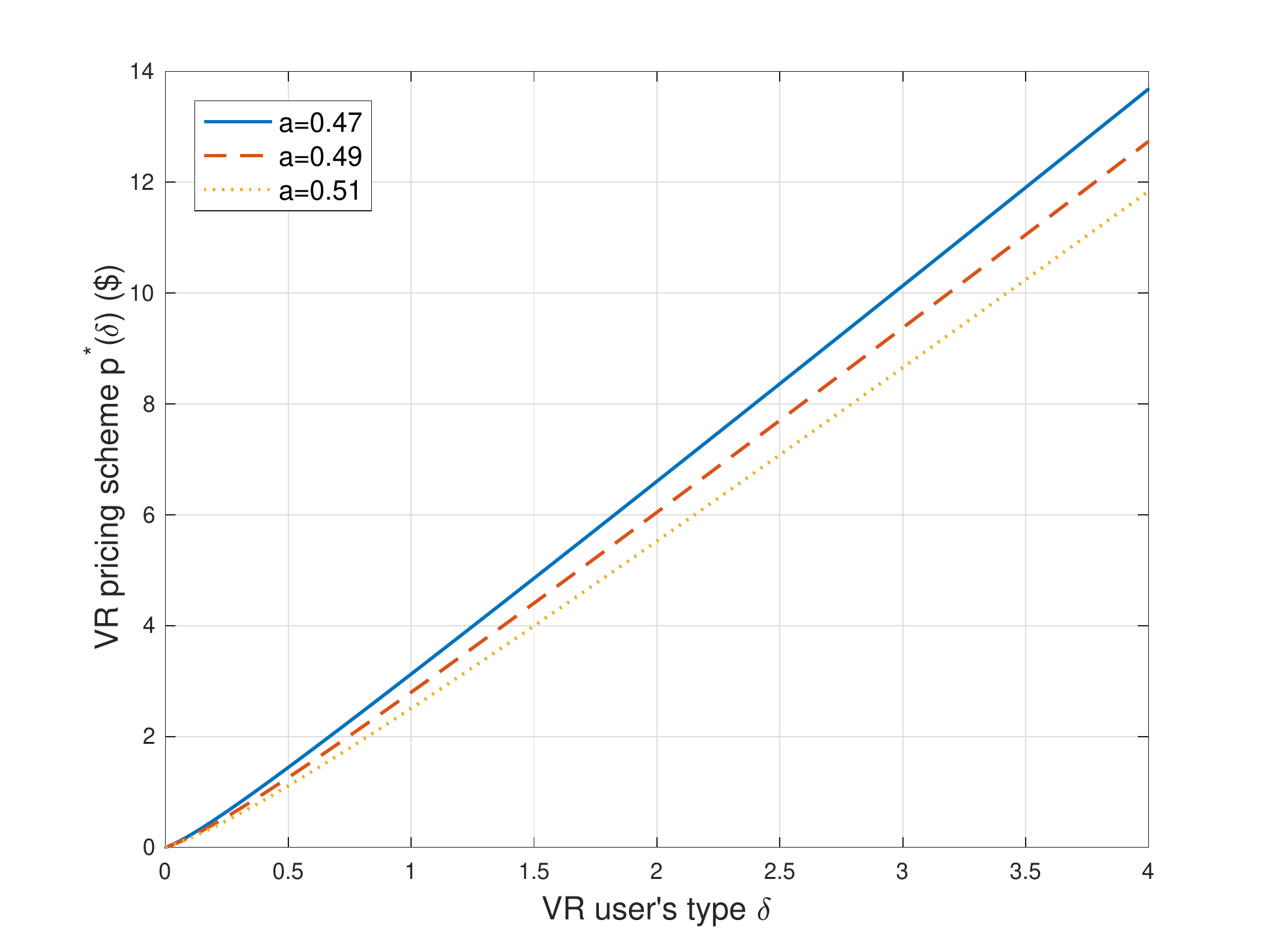}\label{p}}
	 \subfigure[VR QoS $q^*(\delta)$]{
    \includegraphics[width=0.48\columnwidth]{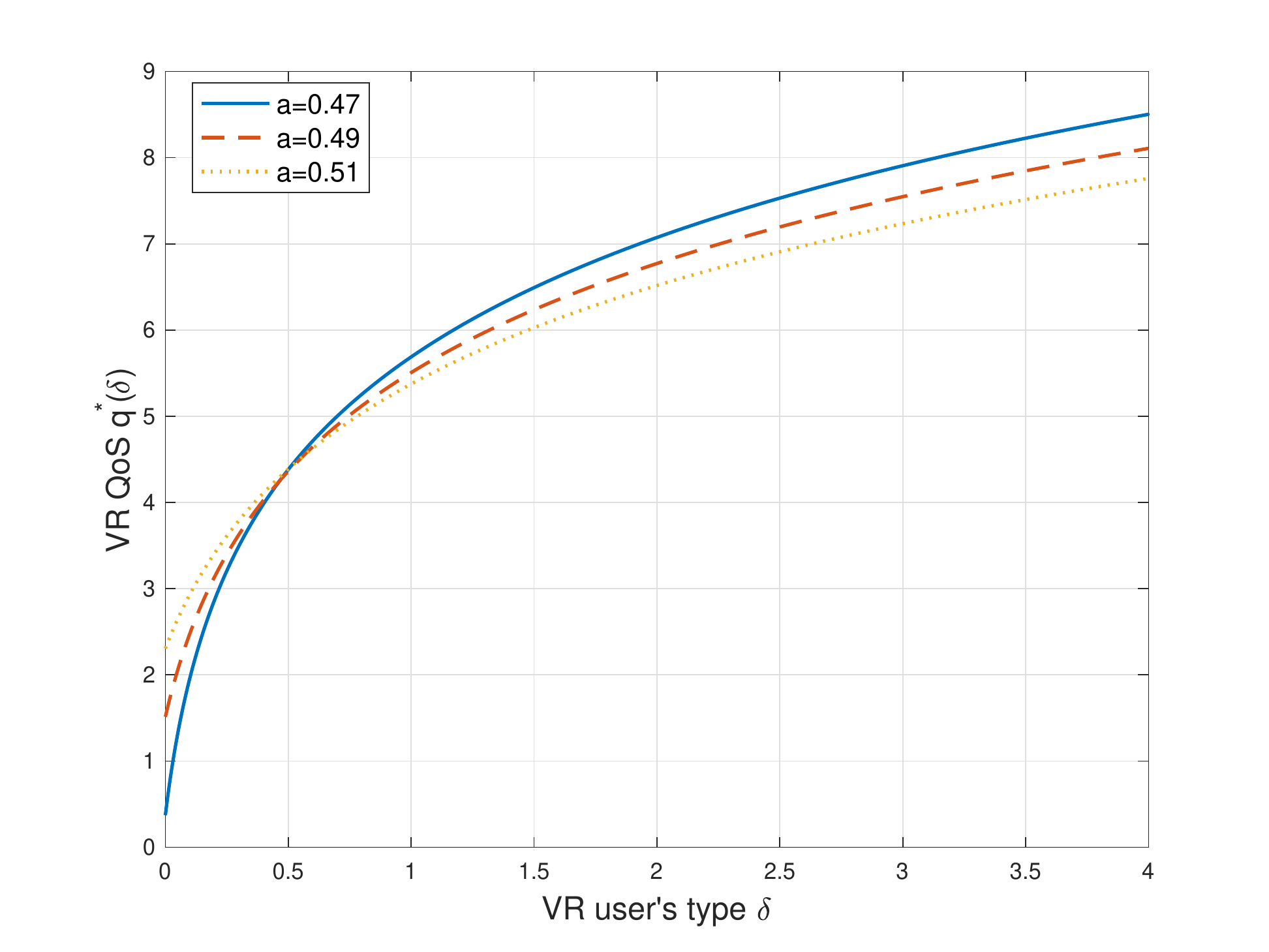}\label{q}}    
     \subfigure[VR QoS in terms of resolution, delay, and reliability]{
    \includegraphics[width=0.95\columnwidth]{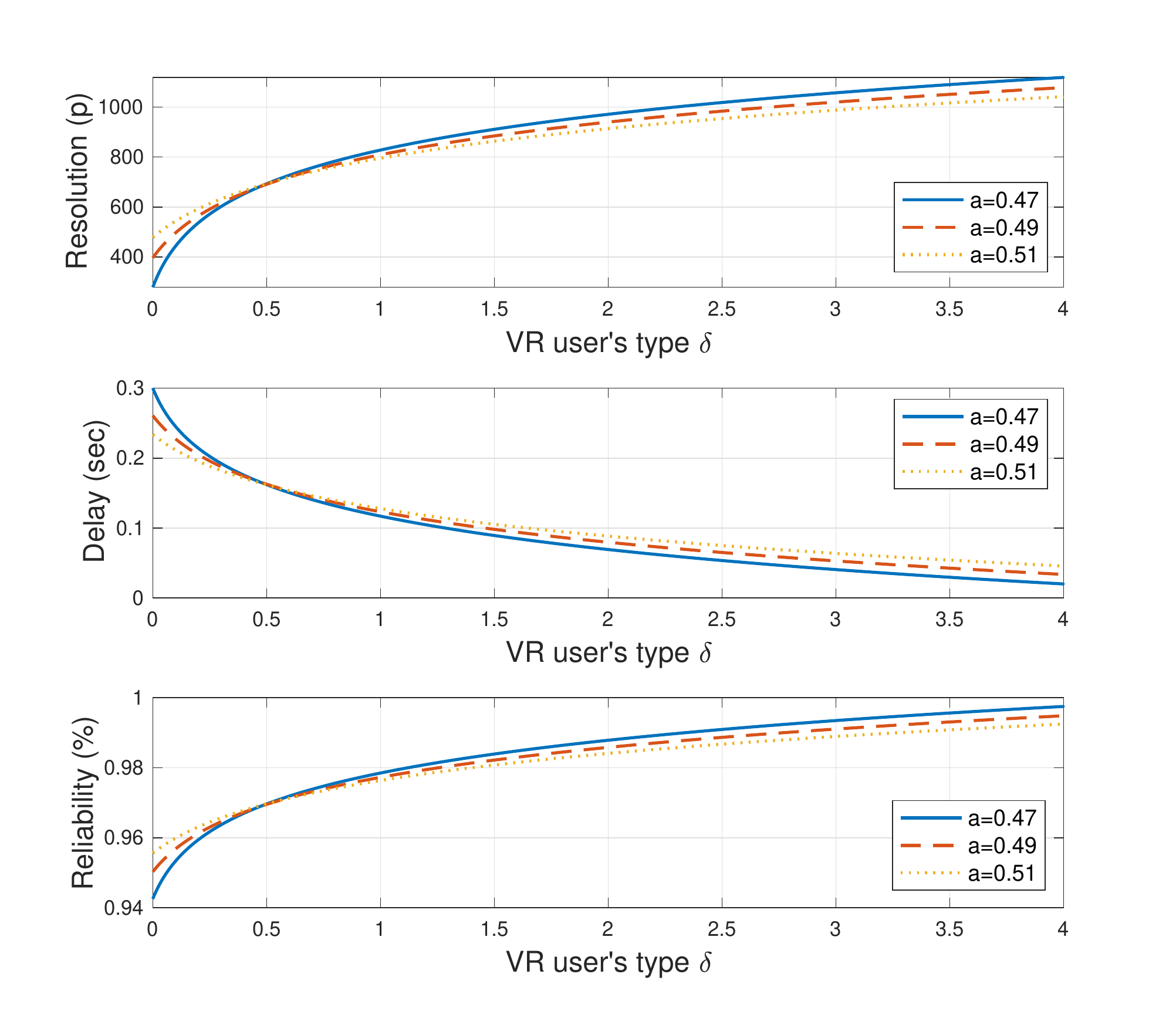}\label{q_metrics}}       
  \caption[]{(a) and (b) illustrate the optimal pricing scheme and the corresponding QoS of VR, respectively. (c) depicts the specific sensing QoS in terms of resolution, delay, and reliability metrics.}
  \label{example}
\end{figure}

\begin{figure}[t]
\begin{centering}
\includegraphics[width=1\columnwidth]{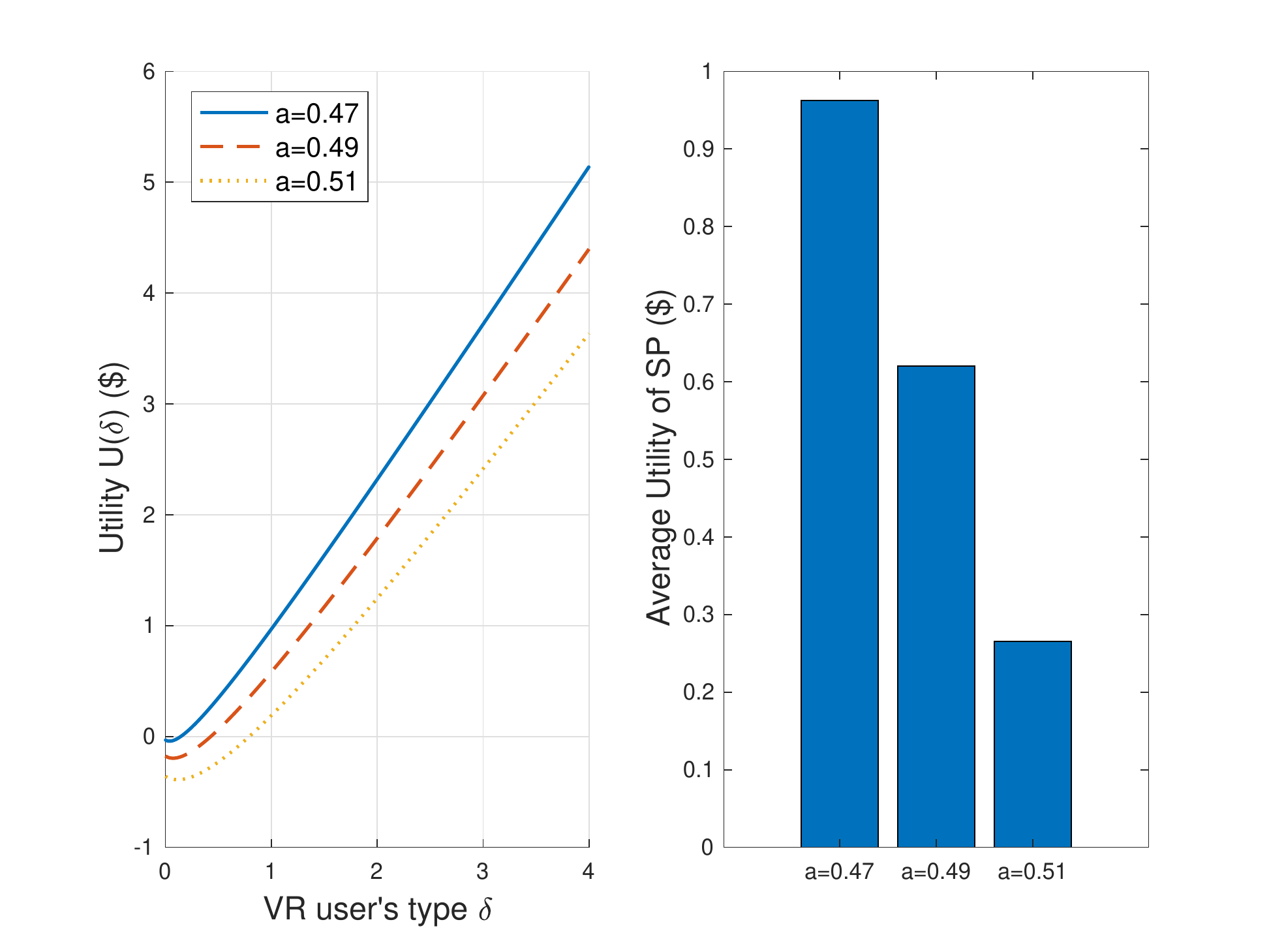}
\par\end{centering}
\caption{\label{case1_utility_SP}
Utility of the SP under hidden information. The SP earns  profits from the users who demand a better VR service.}
\end{figure}

\subsection{Optimal Contracts under Full Information}
For comparison, we present the optimal contracts under the full information based on Theorem \ref{thm2} and quantify the information cost associated with the user's private types. Fig. \ref{benchmark_pq} shows the optimal pricing $p^b(\delta)$ and the QoS mapping $q^b(\delta)$. Specifically, $p^b(\delta)$ is larger than the counterpart $p^*(\delta)$ under asymmetric information. Due to the reputation constraint, the VR QoS $q^b(\delta)$ has a similar trajectory as $q^*(\delta)$.  The corresponding SP's revenue is shown in Fig. \ref{benchmark_utility_SP}. Similarly, a larger $a$ reduces the payoff of the VR SP. Furthermore, we can conclude that the SP earns more by knowing the private user's  type information. For example, when $a=0.47$, the average utility of serving a user is 4.4\$ which is more than 4 times larger than the one under hidden information depicted in Fig. \ref{case1_utility_SP}.

\begin{figure}[t]
\begin{centering}
\includegraphics[width=1\columnwidth]{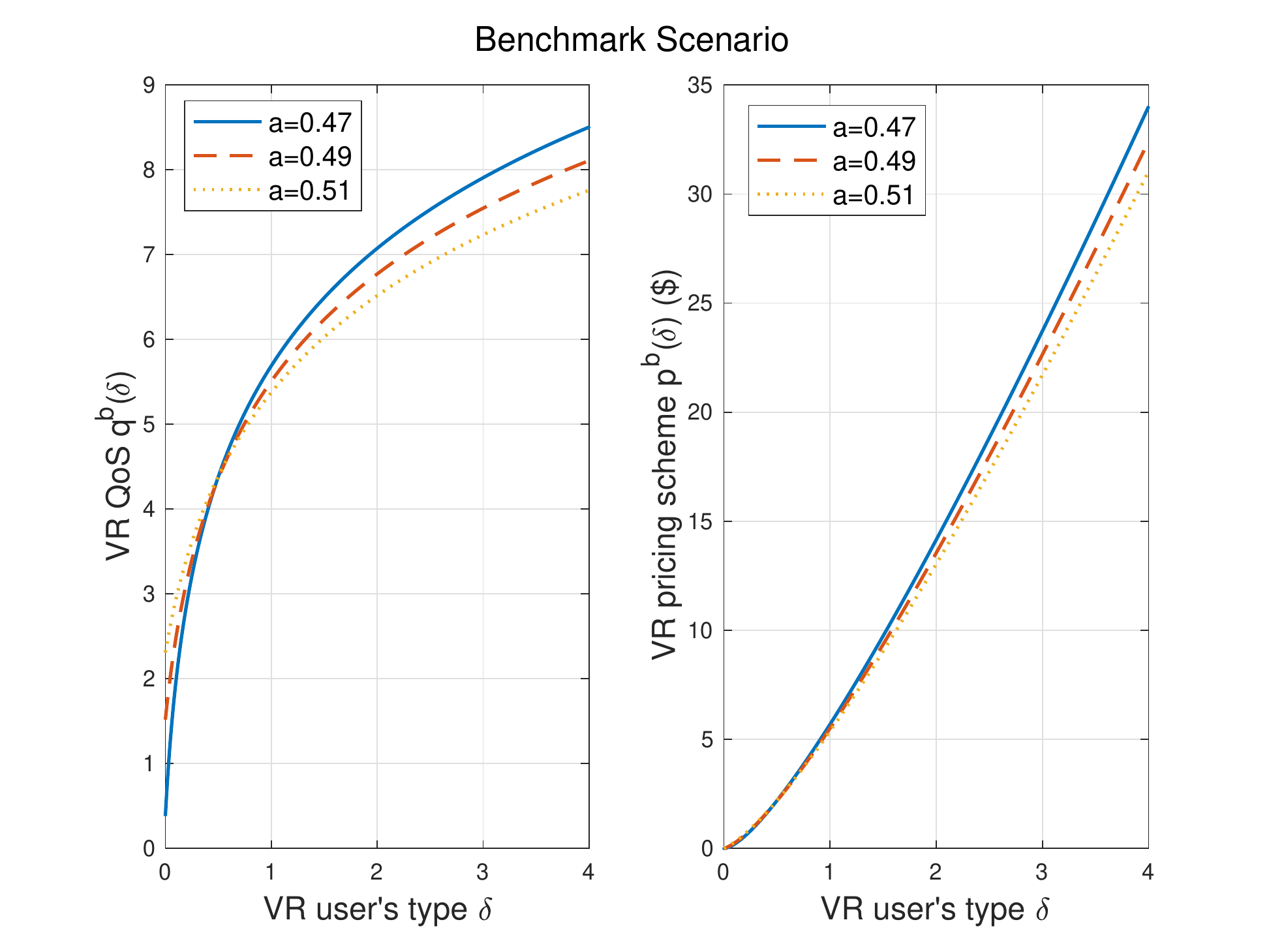}
\par\end{centering}
\caption{\label{benchmark_pq}
Optimal contracts in the benchmark scenario. The VR service pricing $p^b(\delta)$ is larger than the counterpart $p^*(\delta)$.}
\end{figure}

\begin{figure}[t]
\begin{centering}
\includegraphics[width=1\columnwidth]{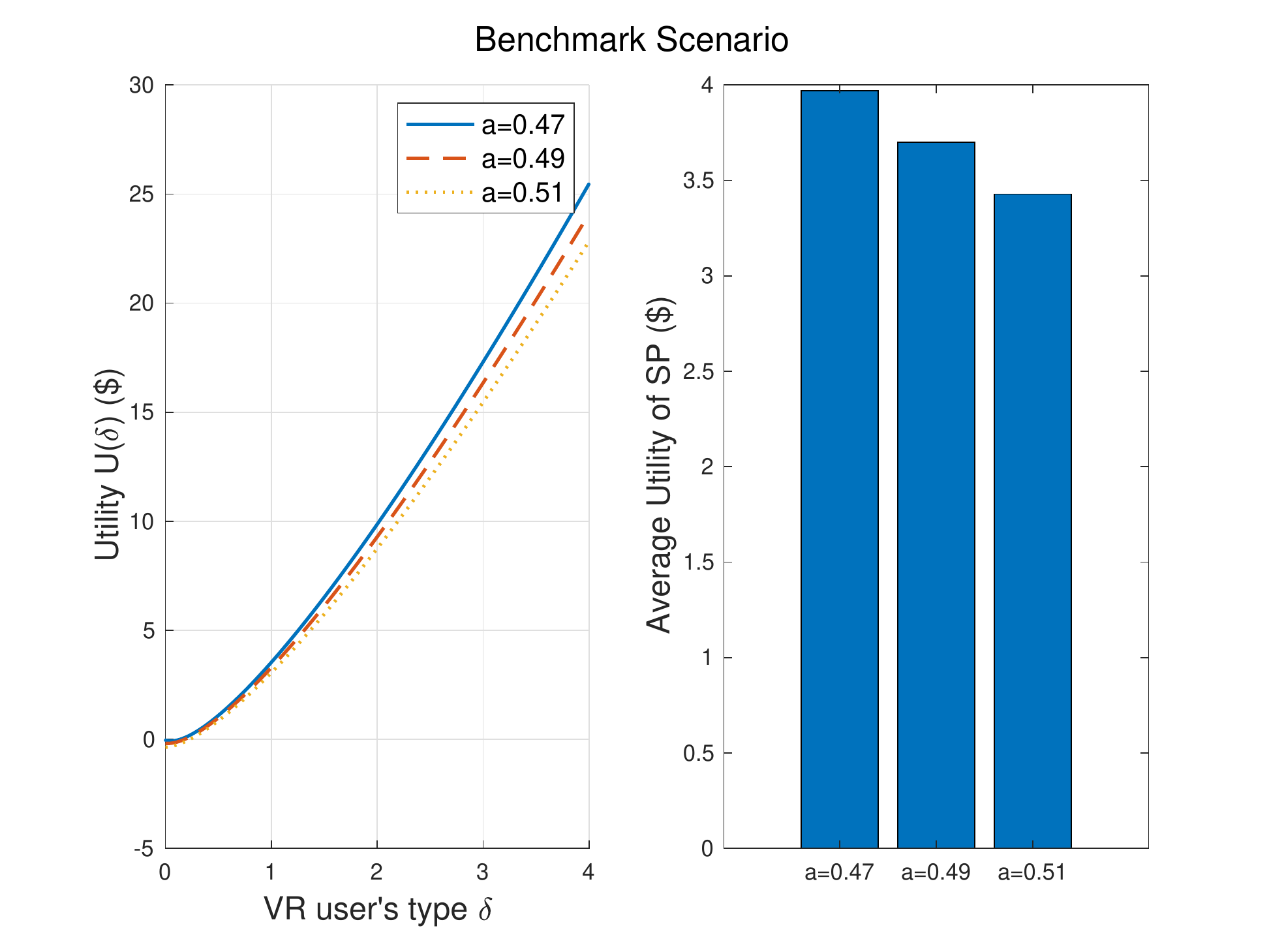}
\par\end{centering}
\caption{\label{benchmark_utility_SP}
Utility of the SP in the benchmark scenario. The SP's revenue under full information is more than 4 times larger than the corresponding one under asymmetric information. }
\end{figure}

\section{Conclusion}\label{conclusion}
In this paper, we have established a 
Sensing-as-Service (SaaS) framework for QoS-based data trading in the IoT markets using contract theory. The proposed framework is designed for massive IoT scenarios where users are characterized by their service requirements and sensing data available to the service provider (SP) is characterized by quality.
Depending on the probability distribution of user's QoS needs, the profit maximizing contract solutions are proposed between the SP and users, which admit different structures.
Specifically, under a wide class of user's type distributions without a large or sudden decrease, the data pricing scheme and QoS mapping are monotonically increasing with the user types. Otherwise, nondiscriminative pricing phenomenon is observed which reduces the diversity of service provisions to the IoT users. Moreover, invariant pricing phenomenon can occur when the user's type distribution decreases exponentially, and thus the service provider targets the majority of users in the market to maximize the profits. 
We have also validated our results using a case study based on the application of the SaaS framework to UAV-enabled virtual reality, where the SP makes more profit by providing data services to higher type users.
Future work can expand the SaaS contract design to cases when bounded rationality is considered in the user's behavior, i.e., users have uncertainty on their type parameters, and subsequently design robust contract mechanisms. Another direction is to develop an online learning approach to designing optimal contract solutions when the user's type distribution is unknown to the SP.

\appendices
\section{Proof of Lemma \ref{L1}} \label{Appx1}
The first-order optimality condition (FOC) on \eqref{V_fun} with respect to $\delta'$ can be expressed as $\frac{\partial{\Phi(\delta,q(\delta'))}}{\partial{q(\delta')}} \frac{dq(\delta')}{d\delta'} - \frac{dp(\delta')}{d\delta'}=0.$ The IC constraint in \eqref{IC} indicates that the user of type $\delta$ achieves the largest payoff when claiming its true type $\delta$. Thus, under $\delta'=\delta$, the FOC becomes $\frac{\partial{\Phi(\delta,q(\delta))}}{\partial{q(\delta)}} \frac{dq(\delta)}{d\delta} - \frac{dp(\delta)}{d\delta}=0,$ which yields the local incentive constraint \eqref{IC_differential}. 
Similarly, the second-order optimality condition (SOC) can be written as:
$
\frac{\partial^2{\Phi(\delta,q(\delta'))}}{\partial{q(\delta')^2}} (\frac{dq(\delta')}{d\delta'})^2 +\frac{\partial{\Phi(\delta,q(\delta'))}}{\partial{q(\delta')}} \frac{d^2q(\delta')}{d\delta'^2} - \frac{d^2p( \delta')}{d\delta'^2}\leq 0.
$ Differentiating \eqref{IC_differential} with respect to $\delta$ further gives
$
\frac{d^2p(\delta)}{d\delta^2} = \frac{\partial{\Phi^2(\delta,q(\delta))}}{\partial{q(\delta)^2}} (\frac{dq(\delta)}{d\delta})^2 +\frac{\partial{\Phi^2(\delta,q(\delta))}}{\partial q(\delta)\partial{\delta}}\frac{dq(\delta)}{d\delta}+
\frac{\partial{\Phi(\delta,q(\delta))}}{\partial{q(\delta)}}\frac{d^2q(\delta)}{d\delta^2},
$
and comparing it with the SOC, we obtain 
$
\frac{\partial{\Phi^2(\delta,q(\delta))}}{\partial q(\delta)\partial{\delta}}\frac{dq(\delta)}{d\delta}\geq 0.
$
Together with Assumption \ref{A1}, we obtain the monotonicity constraint \eqref{monotone}.
The next step is to show that \eqref{IC_differential} and \eqref{monotone} together imply the IC constraint \eqref{IC}. Assume that the IC constraint does not hold for at least one type of users, e.g., $\delta$. Then, there exists a $\tilde{\delta}\neq \delta$ such that $\Phi(\delta,q(\delta))-p(\delta)<\Phi(\delta,q(\tilde\delta))-p(\tilde\delta)$, and hence
$\int_{\delta}^{\tilde{\delta}} (\frac{\partial{\Phi(\delta,q(\tau))}}{\partial{q(\tau)}}\frac{dq(\tau)}{d\tau}-  \frac{dp(\tau)}{d\tau} )d\tau>0
$, where we can check that the derivative of $\Phi(\delta,q(\tau))-p(\tau)$ with respect to $\tau$ is exactly the integrand. 
Then when $\tilde{\delta}>\delta$ which gives $\tau>\delta$, we obtain $\frac{\partial{\Phi(\delta,q(\tau))}}{\partial{q(\tau)}}<\frac{\partial{\Phi(\tau,q(\tau))}}{\partial{q(\tau)}}$ by Assumption \ref{A1}. In addition, \eqref{IC_differential} indicates that $\int_{\delta}^{\tilde{\delta}} (\frac{\partial{\Phi(\tau,q(\tau))}}{\partial{q(\tau)}}\frac{dq(\tau)}{d\tau}-  \frac{dp(\tau)}{d\tau} )d\tau=0$. Replacing $\frac{\partial{\Phi(\tau,q(\tau))}}{\partial{q(\tau)}}$ in the integrand by $\frac{\partial{\Phi(\delta,q(\tau))}}{\partial{q(\tau)}}$ yields
$\int_{\delta}^{\tilde{\delta}} (\frac{\partial{\Phi(\delta,q(\tau))}}{\partial{q(\tau)}}\frac{dq(\tau)}{d\tau}-  \frac{dp(\tau)}{d\tau} )d\tau<0$ since $\frac{\partial{\Phi(\delta,q(\tau))}}{\partial{q(\tau)}}<\frac{\partial{\Phi(\tau,q(\tau))}}{\partial{q(\tau)}}$ and $\frac{dq(\tau)}{d\tau}>0$, and this inequality contradicts with the previous integral inequality. Similar analysis follows for the case when $\tilde{\delta}<\delta$, and we can conclude that \eqref{IC_differential} and \eqref{monotone} imply the IC constraint \eqref{IC}.

\bibliographystyle{IEEEtran}
\bibliography{IEEEabrv,references}

\begin{thebibliography}{10}
\providecommand{\url}[1]{#1}
\csname url@samestyle\endcsname
\providecommand{\newblock}{\relax}
\providecommand{\bibinfo}[2]{#2}
\providecommand{\BIBentrySTDinterwordspacing}{\spaceskip=0pt\relax}
\providecommand{\BIBentryALTinterwordstretchfactor}{4}
\providecommand{\BIBentryALTinterwordspacing}{\spaceskip=\fontdimen2\font plus
\BIBentryALTinterwordstretchfactor\fontdimen3\font minus
  \fontdimen4\font\relax}
\providecommand{\BIBforeignlanguage}[2]{{%
\expandafter\ifx\csname l@#1\endcsname\relax
\typeout{** WARNING: IEEEtran.bst: No hyphenation pattern has been}%
\typeout{** loaded for the language `#1'. Using the pattern for}%
\typeout{** the default language instead.}%
\else
\language=\csname l@#1\endcsname
\fi
#2}}
\providecommand{\BIBdecl}{\relax}
\BIBdecl

\bibitem{zanella2014internet}
A.~Zanella, N.~Bui, A.~Castellani, L.~Vangelista, and M.~Zorzi, ``Internet of
  things for smart cities,'' \emph{IEEE Internet of Things Journal}, vol.~1,
  no.~1, pp. 22--32, 2014.

\bibitem{xiong2017connectivity}
S.~Xiong, Q.~Ni, X.~Wang, and Y.~Su, ``A connectivity enhancement scheme based
  on link transformation in {IoT} sensing networks,'' \emph{IEEE Internet of
  Things Journal}, vol.~4, no.~6, pp. 2297--2308, 2017.

\bibitem{UAV-VR}
Y.~Zhou, C.~Pan, P.~L. Yeoh, K.~Wang, M.~Elkashlan, B.~Vucetic, and Y.~Li,
  ``Communication-and-computing latency minimization for {UAV}-enabled virtual
  reality delivery systems,'' \emph{IEEE Transactions on Communications},
  vol.~69, no.~3, pp. 1723--1735, 2021.

\bibitem{perera2014sensing}
C.~Perera, A.~Zaslavsky, P.~Christen, and D.~Georgakopoulos, ``Sensing as a
  service model for smart cities supported by {I}nternet of things,''
  \emph{Transactions on Emerging Telecommunications Technologies}, vol.~25,
  no.~1, pp. 81--93, 2014.

\bibitem{semasinghe2017game}
P.~Semasinghe, S.~Maghsudi, and E.~Hossain, ``Game theoretic mechanisms for
  resource management in massive wireless {IoT} systems,'' \emph{IEEE
  Communications Magazine}, vol.~55, no.~2, pp. 121--127, 2017.

\bibitem{maximum_principle}
\BIBentryALTinterwordspacing
R.~E. Kopp, ``{Pontryagin} maximum principle,'' in \emph{Optimization
  Techniques}, ser. Mathematics in Science and Engineering, G.~Leitmann,
  Ed.\hskip 1em plus 0.5em minus 0.4em\relax Elsevier, 1962, vol.~5, pp.
  255--279. [Online]. Available:
  \url{https://www.sciencedirect.com/science/article/pii/S0076539208620950}
\BIBentrySTDinterwordspacing

\bibitem{laffont2009theory}
J.-J. Laffont and D.~Martimort, \emph{The Theory of Incentives: The
  Principal-Agent Model}.\hskip 1em plus 0.5em minus 0.4em\relax Princeton
  university press, 2009.

\bibitem{fehr2007fairness}
E.~Fehr, A.~Klein, and K.~M. Schmidt, ``Fairness and contract design,''
  \emph{Econometrica}, vol.~75, no.~1, pp. 121--154, 2007.

\bibitem{doherty1993insurance}
N.~A. Doherty and G.~Dionne, ``Insurance with undiversifiable risk: Contract
  structure and organizational form of insurance firms,'' \emph{Journal of Risk
  and Uncertainty}, vol.~6, no.~2, pp. 187--203, 1993.

\bibitem{corbett2000supplier}
C.~J. Corbett and X.~De~Groote, ``A supplier's optimal quantity discount policy
  under asymmetric information,'' \emph{Management Science}, vol.~46, no.~3,
  pp. 444--450, 2000.

\bibitem{data_market_design_overview}
S.~W. Driessen, G.~Monsieur, and W.-J. Van Den~Heuvel, ``Data market design: A
  systematic literature review,'' \emph{IEEE Access}, vol.~10, pp.
  33\,123--33\,153, 2022.

\bibitem{resource_trading_iot}
S.~Sheng, R.~Chen, P.~Chen, X.~Wang, and L.~Wu, ``Futures-based resource
  trading and fair pricing in real-time {IoT} networks,'' \emph{IEEE Wireless
  Communications Letters}, vol.~9, no.~1, pp. 125--128, 2020.

\bibitem{contracts_resource_trading_mec}
C.~Su, F.~Ye, Y.~Zha, T.~Liu, Y.~Zhang, and Z.~Han, ``Matching with
  contracts-based resource trading and price negotiation in multi-access edge
  computing,'' \emph{IEEE Wireless Communications Letters}, vol.~10, no.~4, pp.
  892--896, 2021.

\bibitem{contract_opportunistic_iot}
N.~Gupta, J.~Singh, S.~K. Dhurandher, and Z.~Han, ``Contract theory based
  incentive design mechanism for opportunistic {IoT} networks,'' \emph{IEEE
  Internet of Things Journal}, pp. 1--1, 2021.

\bibitem{computation_offloading_trans_scheduling_delay_sensitive}
C.~Yi, J.~Cai, and Z.~Su, ``A multi-user mobile computation offloading and
  transmission scheduling mechanism for delay-sensitive applications,''
  \emph{IEEE Transactions on Mobile Computing}, vol.~19, no.~1, pp. 29--43,
  2020.

\bibitem{incentive_mechanism_resource_alloc_edge_fog}
M.~Diamanti, P.~Charatsaris, E.~E. Tsiropoulou, and S.~Papavassiliou,
  ``Incentive mechanism and resource allocation for edge-fog networks driven by
  multi-dimensional contract and game theories,'' \emph{IEEE Open Journal of
  the Communications Society}, vol.~3, pp. 435--452, 2022.

\bibitem{AoI_latency_contract_telco}
X.~Zhou, W.~Wang, N.~U. Hassan, C.~Yuen, and D.~Niyato, ``Towards small aoi and
  low latency via operator content platform: A contract theory-based pricing,''
  \emph{IEEE Transactions on Communications}, vol.~70, no.~1, pp. 366--378,
  2022.

\bibitem{computational_caching_AR_contracts}
T.~N. Dang, K.~Kim, L.~U. Khan, S.~M.~A. Kazmi, Z.~Han, and C.~S. Hong,
  ``On-device computational caching-enabled augmented reality for 5g and
  beyond: A contract-theory-based incentive mechanism,'' \emph{IEEE Internet of
  Things Journal}, vol.~8, no.~24, pp. 17\,382--17\,394, 2021.

\bibitem{optimal_pricing}
K.~Liu, X.~Qiu, W.~Chen, X.~Chen, and Z.~Zheng, ``Optimal pricing mechanism for
  data market in blockchain-enhanced {Internet} of things,'' \emph{IEEE
  Internet of Things Journal}, vol.~6, no.~6, pp. 9748--9761, 2019.

\bibitem{price_auditing_ridehailing}
Y.~Lu, Y.~Qi, S.~Qi, Y.~Li, H.~Song, and Y.~Liu, ``Say no to price
  discrimination: Decentralized and automated incentives for price auditing in
  ride-hailing services,'' \emph{IEEE Transactions on Mobile Computing},
  vol.~21, no.~2, pp. 663--680, 2022.

\bibitem{blockchain_as_a_service_5G}
N.~Weerasinghe, T.~Hewa, M.~Liyanage, S.~S. Kanhere, and M.~Ylianttila, ``A
  novel {Blockchain}-as-a-service ({BaaS}) platform for local 5g operators,''
  \emph{IEEE Open Journal of the Communications Society}, vol.~2, pp. 575--601,
  2021.

\bibitem{data_trading_blockchain_IoT}
L.~D. Nguyen, I.~Leyva-Mayorga, A.~N. Lewis, and P.~Popovski, ``Modeling and
  analysis of data trading on blockchain-based market in {IoT} networks,''
  \emph{IEEE Internet of Things Journal}, vol.~8, no.~8, pp. 6487--6497, 2021.

\bibitem{contract_blockchain_IoT}
J.~Li, T.~Liu, D.~Niyato, P.~Wang, J.~Li, and Z.~Han, ``Contract-theoretic
  pricing for security deposits in sharded blockchain with {Internet} of things
  ({IoT}),'' \emph{IEEE Internet of Things Journal}, vol.~8, no.~12, pp.
  10\,052--10\,070, 2021.

\bibitem{zhang2017survey}
Y.~Zhang, M.~Pan, L.~Song, Z.~Dawy, and Z.~Han, ``A survey of contract
  theory-based incentive mechanism design in wireless networks,'' \emph{IEEE
  Wireless Communications}, vol.~24, no.~3, pp. 80--85, 2017.

\bibitem{zhang2015contract}
Y.~Zhang, L.~Song, W.~Saad, Z.~Dawy, and Z.~Han, ``Contract-based incentive
  mechanisms for device-to-device communications in cellular networks,''
  \emph{IEEE Journal on Selected Areas in Communications}, vol.~33, no.~10, pp.
  2144--2155, 2015.

\bibitem{chen2017promoting}
Y.~Chen, S.~He, F.~Hou, Z.~Shi, and J.~Chen, ``Promoting device-to-device
  communication in cellular networks by contract-based incentive mechanisms,''
  \emph{IEEE Network}, vol.~31, no.~3, pp. 14--20, 2017.

\bibitem{hasan2013relay}
Z.~Hasan and V.~K. Bhargava, ``Relay selection for {OFDM} wireless systems
  under asymmetric information: A contract-theory based approach,'' \emph{IEEE
  Transactions on Wireless Communications}, vol.~12, no.~8, pp. 3824--3837,
  2013.

\bibitem{duan2014cooperative}
L.~Duan, L.~Gao, and J.~Huang, ``Cooperative spectrum sharing: A contract-based
  approach,'' \emph{IEEE Transactions on Mobile Computing}, vol.~13, no.~1, pp.
  174--187, 2014.

\bibitem{gao2011spectrum}
L.~Gao, X.~Wang, Y.~Xu, and Q.~Zhang, ``Spectrum trading in cognitive radio
  networks: A contract-theoretic modeling approach,'' \emph{IEEE Journal on
  Selected Areas in Communications}, vol.~29, no.~4, pp. 843--855, 2011.

\bibitem{chang2018incentive}
Z.~Chang, D.~Zhang, T.~H{\"a}m{\"a}l{\"a}inen, Z.~Han, and T.~Ristaniemi,
  ``Incentive mechanism for resource allocation in wireless virtualized
  networks with multiple infrastructure providers,'' \emph{IEEE Transactions on
  Mobile Computing}, vol.~19, no.~1, pp. 103--115, 2018.

\bibitem{ma2016contract}
Q.~Ma, L.~Gao, Y.-F. Liu, and J.~Huang, ``A contract-based incentive mechanism
  for crowdsourced wireless community networks,'' in \emph{Modeling and
  Optimization in Mobile, Ad Hoc, and Wireless Networks (WiOpt)}, 2016, pp.
  1--8.

\bibitem{nie2018stackelberg}
J.~Nie, J.~Luo, Z.~Xiong, D.~Niyato, and P.~Wang, ``A {Stackelberg} game
  approach toward socially-aware incentive mechanisms for mobile
  crowdsensing,'' \emph{IEEE Transactions on Wireless Communications}, vol.~18,
  no.~1, pp. 724--738, 2018.

\bibitem{xiong2019dynamic}
Z.~Xiong, D.~Niyato, P.~Wang, Z.~Han, and Y.~Zhang, ``Dynamic pricing for
  revenue maximization in mobile social data market with network effects,''
  \emph{IEEE Transactions on Wireless Communications}, vol.~19, no.~3, pp.
  1722--1737, 2020.

\bibitem{duan2014motivating}
L.~Duan, T.~Kubo, K.~Sugiyama, J.~Huang, T.~Hasegawa, and J.~Walrand,
  ``Motivating smartphone collaboration in data acquisition and distributed
  computing,'' \emph{IEEE Transactions on Mobile Computing}, vol.~13, no.~10,
  pp. 2320--2333, 2014.

\bibitem{xiong2020contract}
Z.~Xiong, J.~Zhao, Y.~Zhang, D.~Niyato, and J.~Zhang, ``Contract design in
  hierarchical game for sponsored content service market,'' \emph{IEEE
  Transactions on Mobile Computing}, vol.~20, no.~9, pp. 2763--2778, 2020.

\bibitem{al2017price}
F.~Al-Turjman, ``Price-based data delivery framework for dynamic and pervasive
  {IoT},'' \emph{Pervasive and Mobile Computing}, vol.~42, pp. 299--316, 2017.

\bibitem{al2013priced}
A.~E. Al-Fagih, F.~M. Al-Turjman, W.~M. Alsalih, and H.~S. Hassanein, ``A
  priced public sensing framework for heterogeneous {IoT} architectures,''
  \emph{IEEE Transactions on Emerging Topics in Computing}, vol.~1, no.~1, pp.
  133--147, 2013.

\bibitem{kantarci2014trustworthy}
B.~Kantarci and H.~T. Mouftah, ``Trustworthy sensing for public safety in
  cloud-centric {I}nternet of things,'' \emph{IEEE Internet of Things Journal},
  vol.~1, no.~4, pp. 360--368, 2014.

\bibitem{VR_data}
``How much would you spend on a virtual reality headset?'' \emph{Statista},
  2015, [Online]
  Available:\url{https://www.statista.com/statistics/457117/virtual-reality-headset-amount-willing-to-pay-in-the-united-states/}.

\bibitem{myerson1979incentive}
R.~B. Myerson, ``Incentive compatibility and the bargaining problem,''
  \emph{Econometrica}, pp. 61--73, 1979.

\bibitem{roy2022achieving}
D.~Roy, A.~S. Rao, T.~Alpcan, G.~Das, and M.~Palaniswami, ``Achieving qos for
  bursty urllc applications over passive optical networks,'' \emph{Journal of
  Optical Communications and Networking}, vol.~14, no.~5, pp. 411--425, 2022.

\bibitem{chen2022optimal}
S.~Chen, L.~Wang, and F.~Liu, ``Optimal admission control mechanism design for
  time-sensitive services in edge computing,'' in \emph{IEEE Conference on
  Computer Communications (INFOCOM)}, 2022, pp. 1169--1178.

\bibitem{lakshminarayana2015transmit}
S.~Lakshminarayana, M.~Assaad, and M.~Debbah, ``Transmit power minimization in
  small cell networks under time average qos constraints,'' \emph{IEEE Journal
  on Selected Areas in Communications}, vol.~33, no.~10, pp. 2087--2103, 2015.

\bibitem{kirk2012optimal}
D.~E. Kirk, \emph{Optimal Control Theory: An Introduction}.\hskip 1em plus
  0.5em minus 0.4em\relax Courier Corporation, 2012.

\bibitem{mangasarian1966sufficient}
O.~L. Mangasarian, ``Sufficient conditions for the optimal control of nonlinear
  systems,'' \emph{SIAM Journal on Control}, vol.~4, no.~1, pp. 139--152, 1966.

\bibitem{Tactile_Internet}
G.~Fettweis and et~al, ``{The Tactile Internet: ITU-T Technology Watch
  Report},'' 2014, [Online]
  Available:\url{https://www.itu.int/dms_pub/itu-t/opb/gen/T-GEN-TWATCH-2014-1-PDF-E.pdf}.

\bibitem{video_qos}
Cisco, ``{Video Quality of Service (QOS) Tutorial},'' 2017, [Online]
  Available:\url{https://www.cisco.com/c/en/us/support/docs/quality-of-service-qos/qos-video/212134-Video-Quality-of-Service-QOS-Tutorial.html}.

\end{thebibliography}

\end{document}